\documentclass[11pt]{article}

\usepackage{amsmath,amssymb,amsthm,amsfonts,latexsym,bbm,xspace,graphicx,float,mathtools}
\usepackage[backref,colorlinks,citecolor=blue,bookmarks=true]{hyperref}
\usepackage[letterpaper,margin=1in]{geometry}
\usepackage{color}
\usepackage{comment}
\usepackage{algorithmic,algorithm}

\usepackage{tweaklist}
\renewcommand{\enumhook}{\setlength{\topsep}{2pt}%
  \setlength{\itemsep}{-2pt}}
\renewcommand{\itemhook}{\setlength{\topsep}{2pt}%
  \setlength{\itemsep}{-2pt}}

\usepackage{fullpage}
\usepackage{appendix}
\usepackage[nomessages]{fp}

\newtheorem{theorem}{Theorem}

\newtheorem{corollary}{Corollary}
\newtheorem{lemma}{Lemma}

\newtheorem{claim}{Claim}
\newtheorem{proposition}{Proposition}

\newtheorem{observation}{Observation}

\newcommand{\partition}{\textsc{Partition}}
\newcommand{\select}{\textsc{Select}}
\newcommand{\findMin}{\textsc{findMin}}
\newcommand{\rank}{\textsc{Rank}}

\newenvironment{prevproof}[2]{\noindent {\em {Proof of {#1}~\ref{#2}:}}}{$\Box$\vskip \belowdisplayskip}

\newcommand{\junk}[1]{}

\junk{

}

\newcommand{\poly}{{\rm poly}}

\newcommand{\notshow}[1]{{}}

\DeclareMathOperator{\E}{E}

\definecolor{MyGray}{rgb}{0.8,0.8,0.8}

\begin{document}
\title{Parallel Algorithms for Select and Partition with Noisy Comparisons}

\author{
Mark Braverman \thanks{Department of Computer Science, Princeton University, email: mbraverm@cs.princeton.edu. Research supported in part by an NSF CAREER award (CCF-1149888), NSF CCF-1215990, NSF CCF-1525342, a Packard Fellowship in Science and Engineering, and the Simons Collaboration on Algorithms and Geometry.}
\and
Jieming Mao  \thanks{Department of Computer Science, Princeton University, email: jiemingm@cs.princeton.edu.}
\and 
S. Matthew Weinberg \thanks{Department of Computer Science, Princeton University, email: sethmw@cs.princeton.edu. Research completed in part while the author was a Microsoft Research Fellow at the Simons Institute for the Theory of Computing. }
}
\addtocounter{page}{-1}
\maketitle
\begin{abstract} We consider the problem of finding the $k^{th}$ highest element in a totally ordered set of $n$ elements (\select), and partitioning a totally ordered set into the top $k$ and bottom $n-k$ elements (\partition) using pairwise comparisons. Motivated by settings like peer grading or crowdsourcing, where multiple rounds of interaction are costly and queried comparisons may be inconsistent with the ground truth, we evaluate algorithms based both on their total runtime and the number of interactive rounds in three comparison models: noiseless (where the comparisons are correct), erasure (where comparisons are erased with probability $1-\gamma$), and noisy (where comparisons are correct with probability $1/2+\gamma/2$ and incorrect otherwise). We provide numerous matching upper and lower bounds in all three models. Even our results in the noiseless model, {which is quite well-studied in the TCS literature on parallel algorithms}, are novel.

\end{abstract}

\newpage
\thispagestyle{empty}
\section{Introduction} 
{\emph{Rank aggregation} is a fundamental problem with numerous important applications, ranging from well-studied settings such as social choice~\cite{CaplinN91} and web search~\cite{DworkKNS01} to newer platforms such as crowdsourcing~\cite{ChenBCH13} and peer grading~\cite{PiechHCDNK13}. Salient common features among these applications is that in the end, \emph{ordinal} rather than \emph{cardinal} information about the elements is relevant, and a precise fine-grained ordering of the elements is often unnecessary. For example, the goal of social choice is to select the best alternative, regardless of how good it is. In a curved course, the goal of peer grading is to partition assignments into quantiles corresponding to A/B/C/D, etc, regardless of their absolute quality. 

Prior work has produced numerous ordinal aggregation procedures (i.e. based on comparisons of elements rather than cardinal evaluations of individual elements) in different settings, and we overview those most relevant to our work in Section~\ref{sec:related}. However, existing models from this literature fail to capture an important aspect of the problem with respect to some of the newer applications; that \emph{multiple rounds of interaction are costly}. In crowdsourcing, for instance, one round of interaction is the time it takes to send out a bunch of tasks to users and wait for their responses before deciding which tasks to send out next, which is the main computational bottleneck. In peer grading, each round of interaction might take a week, and grades are expected to be determined certainly within a few weeks. In conference decisions, even one round of interaction seems to be pushing the time constraints.}

Fortunately, the TCS community already provides a vast literature of algorithms with this constraint in mind, under the name of parallel algorithms. For instance, previous work resolves questions like ``how many interactive rounds are necessary for a deterministic or randomized algorithm to select the $k^{th}$ element with $O(n)$ total comparisons?''~\cite{Valiant75,Reischuk81,AjtaiKSS86,AlonA88b,AlonA88a,BollobasB90}. This line of research, however, misses a different important aspect related to these applications (that is, in fact, captured by most works in rank aggregation), that the comparisons might be erroneous. Motivated by applications such as crowdsourcing and peer grading, we therefore study the round complexity of \partition, the problem of partitioning a totally ordered set into the top $k$ and bottom $n-k$ elements, when comparisons might be erroneous.

Our first results on this front provide matching upper and lower bounds on what is achievable for \partition\  in just one round in three different models of error: noiseless (where the comparisons are correct), erasure (where comparisons are erased with probability $1-\gamma$), and noisy (where comparisons are correct with probability $1/2+\gamma/2$ and incorrect otherwise). We provide one-round algorithms using $dn$ comparisons that make $O(n/d), O(n/(d\gamma)), $ and $O(n/(d\gamma^2))$ mistakes (a mistake is any element placed on the wrong side of the partition) with high probability in the three models, respectively. The algorithms are randomized and different for each model, and the bounds hold both when $d$ is an absolute constant or a function of $n$ and $\gamma$. We provide asymptotically matching lower bounds as well: all (potentially randomized) one-round algorithms using $dn$ comparisons necessarily make $\Omega(n/d), \Omega(n/(d\gamma))$, and $\Omega(n/(d\gamma^2))$ mistakes in expectation in the three models, respectively. We further show that the same algorithms and lower bound constructions are also optimal (up to absolute constant factors) if mistakes are instead weighted by various different measures of their distance to $k$, the cutoff.\footnote{Specifically, if $\text{WRONG}_i$ denotes the random variable that is $1$ if an algorithm misplaces $i$ and $0$ otherwise, we consider measures of the following form, for any choice of $c$: $\sum_{i} \text{WRONG}_i |i-k|^c$. For example, $c=0$ counts the number of mistakes. This is further discussed in Section~\ref{sec:prelim}.}

After understanding completely the tradeoff between the number of comparisons and mistakes for one-round algorithms in each of the three models, we turn our attention to multi-round algorithms. Here, the results are more complex and can't be summarized in a few sentences. We briefly overview our multi-round results in each of the three models below. Again, \emph{all} of the upper and lower bounds discussed below extend when mistakes are weighted by their distance to the cutoff. {We overview the techniques used in proving our results in Section~\ref{sec:tools}, but just briefly note here that the level of technicality roughly increases as we go from the noiseless to erasure to noisy models. In particular, lower bounds in the noisy model are quite involved.}

\paragraph{Multi-Round Results in the Noiseless Model.}
\begin{enumerate}
\item We design a 2-round algorithm for \partition \ using $n/\varepsilon$ total comparisons that makes $O(n^{1/2+\varepsilon}\poly(\log n)))$ mistakes with probability {$1-e^{-\Omega(n)}$}, and prove a nearly matching lower bound of $\Omega(\sqrt{n}\cdot \varepsilon^{5/2})$ mistakes, for any $\varepsilon > 0$ ($\varepsilon$ may be a constant or a function of $n$). 
\item We design a 3-round algorithm for \partition\ making $O(n\cdot\poly(\log n))$ total comparisons that makes \emph{zero} mistakes with probability {$1-e^{-\Omega(n)}$}. It is known that $\omega(n)$ total comparisons are necessary for a 3-round algorithm just to solve \select, the problem of \emph{finding} the $k^{th}$ element, {with probability $1-o(1)$}~\cite{BollobasB90}.
\item We design a 4-round algorithm for \partition\ making $O(n)$ total comparisons that makes \emph{zero} mistakes with probability {$1-e^{-\Omega(n)}$}.  This matches the guarantee provided by an algorithm of Bollob\'{a}s and Brightwell for \select, but is significantly simpler {(in particular, it avoids any graph theory)}~\cite{BollobasB90}.
\end{enumerate}

\paragraph{Multi-Round Results in the Erasure Model.}
\begin{enumerate}
\item We design a $O(\log^*(n))$-round algorithm for \partition\ making $O(n/\gamma)$ total comparisons that makes zero mistakes with probability $1-e^{-\Omega(n)}$.
\item We show that no $o(\log^*(n))$-round algorithm even for \select\ making $O(n/\gamma)$ total comparisons can succeed with probability {$2/3$.}
\end{enumerate}

\paragraph{Multi-Round Results in the Noisy Model.}
\begin{enumerate}
\item We design a 4-round algorithm for \partition\ making $O(n\log n /\gamma^2)$ comparisons that makes zero mistakes with high probability (a trivial corollary of our noiseless algorithm). 
\item We show that no algorithm even for \select\ making $o(n\log n /\gamma^2)$ comparisons can succeed with probability $2/3$ (in any number of rounds). 
\item We design an algorithm for \findMin\ (the special case of \select\ with $k = n$) making $O(n/\gamma^2)$ comparisons that succeeds with probability $2/3$. We also show that no algorithm making $o(n\log n/\gamma^2)$ comparisons can solve \findMin\ with probability $1-1/\poly(n)$ (in any number of rounds).
\end{enumerate}

Together, these results tell an interesting story. In one round, one can obtain the same guarantee in the noiseless versus erasure model with an additional factor of $1/\gamma$ comparisons. And one can obtain the same guarantee in the erasure versus noisy model with an additional factor of $1/\gamma$ comparisons. In some sense, this should be expected, because this exactly captures the degradation in information provided by a single comparison in each of the three models (a noiseless comparison provides one bit of information, an erasure comparison provides $\gamma$ bits of information, and a noisy comparison provides $\Theta(\gamma^2)$ bits of information). But in multiple rounds, everything changes. In four rounds, one can perfectly partition with high probability and $O(n)$ total comparisons in the noiseless model. {In the erasure model, one can indeed partition perfectly with high probability and $O(n/\gamma)$ comparisons, but now it requires $\Theta(\log^*(n))$ rounds instead of just $4$. Moreover, in the noisy model, any algorithm even solving \select\ with probability $2/3$ requires an $\Omega(\log n /\gamma)$ blow-up in the number of comparisons, in any number of rounds!} Note that neither of these additional factors come from the desire to succeed with high probability (as the lower bounds hold against even a $2/3$ success) \emph{nor} the desire to partition every element correctly (as the lower bounds hold even for just \select), but just from the way in which interaction helps in the three different models.

While we believe that the story told by our work as a whole provides the ``main result,'' it is also worth emphasizing independently our results in the noisy model. Our one-round algorithm, for instance, is more involved than its counterparts in the noiseless and erasure models and our analysis uses the theory of biased random walks. Our multi-round lower bounds against \select\ and \findMin\ in the noisy model are the most technical results of the paper, and tell their own interesting story about the difference between \findMin\ and \select\ in the noisy model. To our knowledge, most tight lower bounds known for \select\ come directly from lower bounding \findMin. It's surprising that \findMin\ requires $\Theta(\log n)$ fewer comparisons than \select\ to solve with probability $2/3$ in the noisy model. 

{We proceed now by discussing some related works below, and briefly overviewing our techniques in Section~\ref{sec:tools}. We provide some conclusions and future directions in Section~\ref{sec:conclusion}. Our single-round results are discussed in Section~\ref{sec:oneround} and our multi-round results are discussed in Section~\ref{sec:multiround}. However, due to space constraints, all proofs are deferred to the appendix.}

\subsection{Related Work}\label{sec:related}{
Rank aggregation is an enormous field that we can't possibly summarize in its entirety here. Some of the works most related to ours also study \partition\ (sometimes called \textsc{Top}-K). Almost all of these works also consider the possibility of erroneous comparisons, although sometimes under different models where the likelihood of an erroneous comparison scales with the distance between the two compared elements~\cite{ChenS15,Busa-FeketeSCWH13,Eriksson13}. More importantly, to our knowledge this line of work either considers settings where the comparisons are exogenous (the designer has no control over which comparisons are queried, she can just analyze the results), or only analyze the query complexity and not the round complexity of designed algorithms. Our results contribute to this line of work by providing algorithms designed for settings like crowdsourcing or peer grading where the designer does have design freedom, but may be constrained by the number of interactive rounds.

There is a vast literature from the parallel algorithms community studying various sorting and selection problems in the noiseless model. For instance, tight bounds are known on the round complexity of \select\ for deterministic algorithms using $O(n)$ total comparisons (it is $\Theta(\log \log n)$)~\cite{Valiant75,AjtaiKSS86}, and randomized algorithms using $O(n)$ total comparisons (it is $4$)~\cite{AlonA88a,AlonA88b,Reischuk81,BollobasB90}. Similar results are known for sorting and approximate sorting as well~\cite{Cole88,AlonAV86,AjtaiKS83,HaggkvistH81,BollobasT83,BollobasH85,Leighton84}. Many of the designed deterministic algorithms provide \emph{sorting networks}. A sorting network on $n$ elements is a circuit whose gates are binary comparators. The depth of a sorting network is the number of required rounds, and the number of gates is the total number of comparisons. Randomized algorithms are known to require fewer rounds than deterministic ones with the same number of total comparisons for both sorting and selecting~\cite{AlonA88b, BollobasB90}.

In the noisy model, one can of course take any noiseless algorithm and repeat every comparison $O(\log n/\delta^2)$ times in parallel. To our knowledge, positive results that avoid this simple repetition are virtually non-existent. This is likely because a lower bound of Leighton and Ma~\cite{LeightonM00} proves that in fact no sorting network can provide an asymptotic improvement (for complete sorting), and our lower bound (Theorem~\ref{thm:lbnoisy}) shows that no randomized algorithm can provide an asymptotic improvement for \select. To our knowledge, no prior work studies parallel sorting algorithms in the erasure model. } On this front, our work contributes by addressing some open problems in the parallel algorithms literature, but more importantly by providing the first parallel algorithms and lower bounds for \select\ in the erasure and noisy models.

There is also an active study of sorting in the noisy model~\cite{BravermanM08,BravermanM09,MakarychevM13} within the TCS community without concern for parallelization, but with concern for \emph{resampling}. An algorithm is said to resample if it makes the same comparison multiple times. Clearly, an algorithm that doesn't resample can't possibly find the median exactly in the noisy model (what if the comparison between $n/2$ and $n/2+1$ is corrupted?). The focus of these works is designing poly-time algorithms to find the maximum-likelihood ordering from a set of $\binom{n}{2}$ noisy comparisons. Our work is fundamentally different from these, as we have asymptotically fewer than $\binom{n}{2}$ comparisons to work with, and at no point do we try to find a maximum-likelihood ordering (because we only want to solve \partition).
\subsection{Tools and Techniques}\label{sec:tools}

\paragraph{Single Round Algorithms and Lower Bounds.} Our single round results are guided by the following surprisingly useful observation: in order for an algorithm to possibly know that $i$ exceeds the $k^{th}$ highest element, $i$ must at least be compared to some element between itself and $k$ (as otherwise, the comparison results would be identical if we replaced $i$ with an element just below $k$). Unsurprisingly, it is difficult to guarantee that many elements within $n/d$ of $k$ are compared to elements between themselves and $k$ using only $dn$ total comparisons in a single round, and this forms the basis for our lower bounds. Our upper bounds make use of this observation as well, and basically are able to guarantee that an element is correctly placed with high probability whenever it is compared to an element between itself and $k$. It's interesting that the same intuition is key to both the upper and lower bounds. We provide a description of the algorithms and proofs in Section~\ref{sec:oneround}.

In the erasure model, the same intuition extends, except that in order to have a non-erased comparison between $i$ and an element between $i$ and $k$, we need to make roughly $1/\gamma$ such comparisons. This causes our lower bounds to improve by a factor of $1/\gamma$. In the noisy model, the same intuition again extends, although this time the right language is that we need to learn $\Omega(1)$ bits of information from comparisons of $i$ to elements between $i$ and $k$, which requires $\Omega(1/\gamma^2)$ such comparisons, and causes the improved factor of $1/\gamma^2$ in our lower bounds. Our algorithms in these two models are similar to the noiseless algorithm, but the analysis becomes necessarily more involved. For instance, our analysis in the noisy model appeals to facts about biased random walks on the line.

\paragraph{Multi-Round Algorithms and Lower Bounds.} Our constant-round algorithms in the noiseless model are based on the following intuition: once we reach the point that we are only uncertain about $o(n)$ elements, we are basically looking at a fresh instance of \partition\ on a significantly smaller input size, except we're still allowed $\Theta(n)$ comparisons per round. Once we're only uncertain about only $O(\sqrt{n})$ elements, one additional round suffices to finish up (by comparing each element to every other one). The challenge in obtaining a four-round algorithm (as opposed to just an $O(1)$-round algorithm) is ensuring that we make significant enough gains in the first three rounds. 

Interestingly, these ideas for constant-round algorithms in the noiseless model don't prove useful in the erasure or noisy models. Essentially the issue is that even after a constant number of rounds, we are unlikely to be confident that many elements are above or below $k$, so we can't simply recurse on a smaller instance. Still, it is quite difficult to discover a formal barrier, so our multi-round lower bounds for the erasure and noisy models are quite involved. We refer the reader to Section~\ref{sec:multiround} for further details.

\subsection{Conclusions}\label{sec:conclusion}
We study the problems of \partition\ and \select\ in settings where interaction is costly in the noiseless, erasure, and noisy comparison models. We provide matching (up to absolute constant factors) upper and lower bounds for one round algorithms in all three models, which also show that the number of comparisons required for the same guarantee degrade proportional to the information provided by a single comparison. We also provide matching upper and lower bounds for multi-round algorithms in all three models, which also show that the round and query complexity required for the same guarantee in these settings degrades worse than just by the loss in information when moving between the three comparison models. Finally, we show a separation between \findMin\ and \select\ in the noisy model. 

We believe our work motivates two important directions for future work. First, our work considers some of the more important constraints imposed on rank aggregation algorithms in applications like crowdsourcing or peer grading, but not all. For instance, some settings might require that every submission receives the same amount of attention (i.e. is a member of the same number of comparisons), or might motivate a different model of error (perhaps where mistakes aren't independent or identical across comparisons). It would be interesting to design algorithms and prove lower bounds under additional restrictions motivated by applications. 

Finally, it is important to consider incentives in these applications. In peer grading, for instance, the students themselves are the ones providing the comparisons. An improperly designed algorithm might provide ``mechanism design-type'' incentives for the students to actively misreport if they think it will boost their own grade. Additionally, there are also ``scoring rule-type'' incentives that come into play: grading assignments takes effort! Without proper incentives, students may choose to put zero or little effort into their grading and just provide random information. We believe that using ordinal instead of cardinal information will be especially helpful on this front, as it is much easier to design mechanisms when players just make binary decisions, and it's much easier to understand how the noisy information provided by students scale with effort (in our models, it is simply that $\gamma$ will increase with effort). It is therefore important to design mechanisms for applications like peer grading by building off of our algorithms.


\section{Preliminaries and Notation}
\label{sec:prelim}
In this work, we study two problems, \select\ and \partition. Both problems take as input a randomly sorted, totally ordered set and an integer $k$. For simplicity of notation, we denote the $i^{th}$ smallest element of the set as $i$. So if the input set is of size $n$, the input is exactly $[n]$. In \select, the goal is to output the (location of the) element $k$. In \partition, the goal is to partition the elements into the top $k$, which we'll call $A$ for Accept and the bottom $n-k$, which we'll call $R$ for Reject. Also for ease of notation, we'll state all of our results for $k = n/2$, the median, w.l.o.g.\footnote{We show formally in Appendix~\ref{app:technical} that this is indeed w.l.o.g.}

We say an algorithm solves \select\ if it outputs the median, and solves \partition\ if it places correctly all elements above and below the median. For \select, we will say that an algorithm is a $t$-approximation with probability $p$ if it outputs an element in $[n/2-t,n/2+t]$ with probability at least $p$. For \partition, we will consider a class of success measures, parameterized by a constant $c$, and say the $c$-weighted error associated with a specific partitioning into $A \sqcup R$ is equal to $\sum_{i > n/2} I(i \in R)i^c + \sum_{i < n/2} I(i \in A)i^c$.\footnote{For instance, $c = 0$ counts the number of mistakes. $c=1$ counts the number of mistakes, weighted by the distance of the mistaken element from the median. $c=2$ is similar to mean-squared-error, etc.} Interestingly, in all cases we study, the same algorithm is asymptotically optimal for all $c$.

\paragraph{Query and Round Complexity.} Our algorithms will be comparison-based. We study both the number of queries, and the number of adaptive rounds necessary to achieve a certain guarantee.\footnote{For example, an algorithm that makes $Q$ queries one at a time, waiting for the result of previous queries before deciding which queries to make next has round complexity $Q$. An algorithm that makes all queries up front, without knowing any results has round complexity $1$. We call protocols with round complexity $1$ \emph{non-adaptive}.} We may not always emphasize the runtime of our algorithms, but they all run in time $\poly(n)$.

\paragraph{Notation.} We always consider settings where the input elements are a priori indistinguishable, or alternatively, that our algorithms randomly permute the input before making comparisons. When we write $x < y$, we mean literally that $x < y$ in the ground truth. In the noisy model, the results of comparisons may disagree with the underlying ordering, so we say that $x$ beats $y$ if a noisy comparison of $x$ and $y$ returned $x$ as larger than $y$ (regardless of whether or not $x >y$).

\paragraph{Models of Noise.} We consider three comparison models, which return the following when $a > b$.
\begin{itemize}
\item Noiseless: Returns $a$ beats $b$.
\item Erasure: Returns $a$ beats $b$ with probability $\gamma$, and $\bot$ with probability $1-\gamma$. 
\item Noisy: Returns $a$ beats $b$ with probability $1/2 + \gamma/2$, and $b$ beats $a$ with probability $1/2 - \gamma/2$.
\end{itemize}

\paragraph{\partition\ versus \select.} We design all of our algorithms for \partition, and prove all of our lower bounds against \select. We do this because \select\ is in some sense a strictly easier problem than \partition. We discuss how one can get algorithms for \select\ via algorithms for \partition\ and vice versa formally in Appendix~\ref{app:technical}.

\paragraph{Resampling.} Finally, note that in the erasure and noisy models, it may be desireable to query the same comparison multiple times. This is called \emph{resampling}. It is easy to see that without resampling, it is impossible to guarantee that the exact median is found with high probability, even when all $\binom{n}{2}$ comparisons are made (what if the comparison between $n/2$ and $n/2+1$ is corrupted?). Resampling is not necessarily undesireable in the applications that motivate this work, so we consider our main results to be in the model where resampling is allowed. Still, it turns out that all of our algorithms can be easily modified to avoid resampling at the (necessary) cost of a small additional error, and it is easy to see the required modifications.\footnote{Essentially, replace all resampled comparisons with comparisons to ``nearby'' elements instead.} All of our lower bounds hold even against algorithms that resample. 
\section{Results for Non-Adaptive Algorithms}\label{sec:oneround}
In this section, we provide our results on non-adaptive (round complexity = 1) algorithms. We begin with the upper bounds below, followed by our matching (up to constant factors) lower bounds.
\subsection{Upper Bounds}
We provide asymptotically optimal algorithms in each of the three comparison models. Our three algorithms actually choose the same comparisons to make, but determine whether or not to accept or reject an element based on the resulting comparisons differently. The algorithms pick a \emph{skeleton set} $S$ of size $\sqrt{n}$ and compare every element in $S$ to every other element. Each element not in $S$ is compared to $d-1$ random elements of $S$. Pseudocode for this procedure is given in Appendix~\ref{app:oneround}.

From here, the remaining task in all three models is similar: the algorithm must first estimate the rank of each element in the skeleton set. Then, for each $i$, it must use this information combined with the results of $d-1$ comparisons to guess whether $i$ should be accepted or rejected. The correct approach differs in the three models, which we discuss next. 

\paragraph{Noiseless Model.} Pseudocode for our algorithm in the noiseless model is provided as Algorithm~\ref{alg:nanoiseless} in Appendix~\ref{app:oneround}. First, we estimate that the median of the skeleton set, $x$, is close to the actual median. Then, we hope that each $i \notin S$ is compared to some element in $S$ between itself and $x$. If this happens, we can pretty confidently accept or reject $i$. If it doesn't, then all we learn is that $i$ is beaten by some elements above $x$ and it beats some elements below $x$, which provides no helpful information about whether $i$ is above or below the median, so we just make a random decision. 

\begin{theorem}\label{thm:nanoiseless} {Algorithm~\ref{alg:nanoiseless} has query complexity $dn$, round complexity $1$, does not resample, and outputs a partition that, for all $c$, has:
\begin{itemize}
\item expected $c$-weighted error $O((n/d)^{c+1})$, for any $d = o(n^{1/4})$
\item $c$-weighted error $O((n/d)^{c+1})$ with probability $1- e^{-\Omega(n^3/d^{2c+2})}$, for any $d = o(n^{1/4})$.
\end{itemize}}
\end{theorem}

We provide a complete proof of Theorem~\ref{thm:nanoiseless} in Appendix~\ref{app:oneround}. The main ideas are the following. There are two sources of potential error in Algorithm~\ref{alg:nanoiseless}. First, maybe the skeleton set is poorly chosen and not representative of the ground set. But this is extremely unlikely with such a large skeleton set. Second, note that if $i$ is compared to any element in $S$ between itself and $x$, \emph{and} $x$ is very close to $n/2$, then $i$ will be correctly placed. If $|i - n/2| > n/d$, then we're unlikely to miss this window on $d-1$ independent tries, and $i$ will be correctly placed.

\paragraph{Erasure Model.}
In the erasure model, pseudocode for the complete algorithm we use is Algorithm~\ref{alg:naerasure} in Appendix~\ref{app:oneround}. At a high level, the algorithm is similar to Algorithm~\ref{alg:nanoiseless} for the noiseless model, so we refer the reader to Appendix~\ref{app:oneround} to see the necessary changes.

\begin{theorem}\label{thm:naerasure} Algorithm~\ref{alg:naerasure} has query complexity $dn$, round complexity $1$, does not resample, and outputs a partition that, for all $c$, has:
\begin{itemize}
\item expected $c$-weighted error $O((n/(d\gamma))^{c+1})$, for any $d, \gamma$ such that $d/\gamma = o(n^{1/4})$
\item $c$-weighted error $O((n/(d\gamma))^{c+1})$ with probability $1-e^{-\Omega(n^3/d^{2c+2})}$, whenever $d/\gamma = o(\sqrt{n})$ and $d\gamma = o(n^{1/4})$.
\end{itemize}
\end{theorem}

We again postpone a complete proof of Theorem~\ref{thm:naerasure} to Appendix~\ref{app:oneround}. The additional ingredient beyond the noiseless case is a proof that with high probability, not too many of the comparisons within $S$ are erased and therefore while we can't learn the median of $S$ exactly, we can learn a set of almost $|S|/2$ elements that are certainly above the median, and almost $|S|/2$ elements that are certainly below. If $i \notin S$ beats an element that is certainly above the median of $S$, we can confidently accept it, just like in the noiseless case.
\paragraph{Noisy Model.}
Pseudocode for our algorithm in the noisy model is provided as Algorithm~\ref{alg:nanoisy} in Appendix~\ref{app:oneround}. Algorithm~\ref{alg:nanoisy} is necessarily more involved than the previous two. We can still recover a good ranking of the elements in the skeleton set using the Braverman-Mossel algorithm~\cite{BravermanM08}, so this isn't the issue. The big difference between the noisy model and the previous two is that no single comparison can guarantee that $i\notin S$ should be accepted or rejected. Instead, every time we have a set of elements all above the median of $S$, $x$, of which $i$ beats at least half, this provides some evidence that $i$ should be accepted. Every time we have a set of elements all below $x$ of which $i$ is beaten by at least half, this provides some evidence that $i$ should be rejected. The trick is now just deciding which evidence is stronger. Due to space constraints, we refer the reader to Algorithm~\ref{alg:nanoisy} to see our algorithm, which we analyze using theory from biased random walks on the line.

\begin{theorem}\label{thm:nanoisy} Algorithm~\ref{alg:nanoisy} has query complexity $dn$, round complexity $1$, does not resample, and outputs a partition that, for all $c$, has:
\begin{itemize}
\item expected $c$-weighted error $O((n/(d\gamma^2))^{c+1})$, for any $d = o(n^{1/4})$, $\gamma = \omega(n^{1/8})$. 
\item $c$-weighted error $O((n/(d\gamma^2))^{c+1})$ with probability $1-e^{\Omega(n^3/d^{2c+2})}$, for any $d = o(n^{1/4})$, $\gamma = \omega(n^{1/8})$.
\end{itemize}
\end{theorem}

\subsection{Lower Bounds}
In this section, we show that the algorithms designed in the previous section are optimal up to constant factors. All of the algorithms in the previous section are ``tight,'' in the sense that we expect element $i$ to be correctly placed whenever it is compared to enough elements between itself and the median. In the noiseless model, one element is enough. In the erasure model, we instead need $\Omega(1/\gamma)$ (to make sure at least one isn't erased). In the noisy model, we need $\Omega(1/\gamma^2)$ (to make sure we get $\Omega(1)$ bits of information about the difference between $i$ and the median). If we don't have enough comparisons between $i$ and elements between itself and the median, we shouldn't hope to be able to classify $i$ correctly, as the comparisons involving $i$ would look nearly identical if we replaced $i$ with an element just on the other side of the median. Our lower bounds capture this intuition formally, and are all proved in Appendix~\ref{app:oneround}.

\begin{theorem}\label{thm:nalowerbounds} For all $c$, {$d>0$}, any non-adaptive algorithm with query complexity $dn$ necessarily has expected $c$-weighted error $\Omega((n/d)^{c+1})$ in the noiseless model, $\Omega((n/(d\gamma))^{c+1})$ in the erasure model, and $\Omega((n/(d\gamma^2))^{c+1})$ in the noisy model.
\end{theorem}

\section{Results for Multi-Round Algorithms}\label{sec:multiround}

\subsection{Noiseless Model}
We first present our algorithm and nearly matching lower bound for 2-round algorithms. The first round of our algorithm tries to get as good of an approximation to the median as possible, and then compares it to every element in round two. Getting the best possible approximation is actually a bit tricky. For instance, simply finding the median of a skeleton set of size $\sqrt{n}$ only guarantees an element within $\Theta(n^{3/4})$ of the median.\footnote{{This is exactly what Bollob\'{a}s and Brightwell do in the first round of their 4-round algorithm, which is why sophisticated graph theory follows to fit into four rounds. Our improved first round simplifies the remaining rounds.}} We instead take several ``iterations'' of nested skeleton sets to get a better and better approximation to the median. In reality, all iterations happen simultaneously in the first round, but it is helpful to think of them as sequential refinements.

For any $r \geq 1$, our algorithm starts with a huge skeleton set $S_1$ of $n^{2r/(2r+1)}$ random samples from $[n]$. This is too large to compare every element in $S_1$ with itself, so we choose a set $T_1\subseteq S_1$ of $n^{1/(2r+1)}$ random pivots. Then we compare every element in $S_1$ to every element in $T_1$, and we will certainly learn two pivots, $a_1$ and $b_1$ such that the median of $S_1$ lies in $[a_1,b_1]$, and a $p_1$ such that the median of $S_1$ is exactly the $(p_1 |A_1|)^{th}$ element of $A_1 = S_1 \cap [a_1,b_1]$. Now, we recurse within $A_1$ and try to find the $(p_1 |A_1|)^{th}$ element. Of course, because all of these comparisons happen in one round, we don't know ahead of time in which subinterval of $S_1$ we'll want to recurse, so we have to waste a bunch of comparisons. These continual refinements still make some progress, and allow us to find a smaller and smaller window containing the median of $S_1$, which is a very good approximation to the true median because $S_1$ was so large. Pseudocode for our algorithm is Algorithm~\ref{alg:iter2round} in Appendix~\ref{app:multinoiseless}, which ``recursively'' tries to find the $(p_i |A_i|)^{th}$ element of $A_i$.

\begin{theorem}\label{thm:iter2round} For all $c,r$ and $\varepsilon > 0$, Algorithm~\ref{alg:iter2round} has round complexity $2$, query complexity $(r+1)n$, and outputs a partition that:{
\begin{itemize}
\item has expected $c$-weighted error at most $(8rn^{(r+1)/(2r+1)+\varepsilon})^{c+1}$
\item has $c$-weighted error at most $(8rn^{(r+1)/(2r+1)+\varepsilon})^{c+1}$ with probability at least $1-re^{-n^{\Omega(\varepsilon)}}$.
\end{itemize}}

Note that setting $r = \log n$, and $\varepsilon$ such that $n^{\varepsilon} = 8\log^3 n$, we get an algorithm with round complexity $2$, query complexity $n\log n + n$ that outputs a partition with $c$-weighted error $O((\sqrt{n}\log^4 n)^{c+1})$ with probability $1-O(\log n /n^2)$.

\end{theorem}

We also prove a nearly matching lower bound on two-round algorithms in the noiseless model. At a \emph{very} high level, our lower bound repeats the argument of our one round lower bound twice. Specifically, we show that after one round, there are many elements within a window of size $\Theta(n/d)$ of the median such that a constant fraction of these elements have not been compared to any other elements in this window. We then show that after the second round, conditioned on this, there is necessarily a window of size $\approx\sqrt{n}$ such that a constant fraction of these elements have not been compared to any other elements in this window. Finally we show that this implies that we must err on a constant fraction of these elements. The actual proof is technical, but follows this high level outline. Proofs of Theorems~\ref{thm:iter2round} and~\ref{thm:lowerbound2round} can be found in Appendix~\ref{app:multinoiseless}.

\begin{theorem}\label{thm:lowerbound2round}
For all $c$, and any $d = o(n^{1/5})$, any algorithm with query complexity $dn$ and round complexity $2$ necessarily has expected $c$-weighted error $\Omega((\sqrt{n}/d^{5/2})^{c+1})$.
\end{theorem}
From here we show how to make use of our two-round algorithm to design a three-round algorithm that makes \emph{zero} mistakes with high probability. After our two-round algorithm with appropriate parameters, we can be pretty sure that the median lies somewhere in a range of $O(\sqrt{n}\log^4n)$, so we can just compare all of these elements to each other in one additional round. Pseudocode for Algorithm~\ref{alg:3round} is in Appendix~\ref{app:multinoiseless}.

\begin{theorem}\label{thm:3round}
For all $c$, Algorithm~\ref{alg:3round} has query complexity $O(n\log^8 n)$, round complexity $3$, and outputs a partition with \emph{zero} $c$-weighted error with probability $1-O(\log n/n^2)$. 
\end{theorem}

Again, recall that $\omega(n)$ queries are necessary for any three-round algorithm just to solve \select\ with probability $1-o(1)$~\cite{BollobasB90}. Finally, we further make use of ideas from our two-round algorithm to design a simple four round algorithm that has query complexity $O(n)$ and makes \emph{zero} mistakes with high probability. More specifically, we appropriately tune the parameters for our two-round algorithm (i.e. set $r=1$) to find a window of size $\approx n^{2/3}$ that contains the median (and already correctly partition all other elements). We then use similar ideas in round three to further find a window of size $\approx \sqrt{n}$ that contains the median (and again correctly partition all other elements). We use the final round to compare all remaining uncertain elements to each other and correctly partition them.

\begin{theorem}\label{thm:4rounds}
For all $c$, and any $\varepsilon \in (0,1/18)$, Algorithm~\ref{alg:4rounds} has query complexity $O(n)$, round complexity $4$, and outputs a partition with \emph{zero} $c$-weighted error with probability at least $1 - e^{-\Omega(n^{\varepsilon})}$. 
\end{theorem}
\subsection{Erasure and Noisy Models}
Here we briefly overview our results on multi-round algorithms in the erasure and noisy models. We begin with an easy reduction from these models to the noiseless model, at the cost of a blow-up in the round or query complexity. Essentially, we are just observing that one can adaptively resample any comparison in the erasure model until it isn't erased (which will take $1/\gamma$ resamples in expectation), and also that one can resample in parallel any comparison in either the erasure or noisy model the appropriate number of times and have it effectively be a noiseless comparison.

\begin{proposition}
\label{pro:blowup}
If there is an algorithm solving \partition, \select\ or \findMin\ in the noiseless model with probability $p$ that has query complexity $Q$ and round complexity $r$, then there are also algorithms that resample that:
\begin{itemize}
\item solve  \partition, \select\ or \findMin\  in the erasure model with probability $p$ that have expected query complexity $Q/\gamma$, but perhaps with expected round complexity $Q/\gamma$ as well.
\item solve  \partition, \select\ or \findMin\ in the erasure model with probability $p-1/\poly(n)$ that have query complexity $O(Q(\log Q + \log n)/\gamma)$, and round complexity $r$.
\item solve  \partition, \select\ or \findMin\ in the noisy model with probability $p - 1/\poly(n)$ that have query complexity $O(Q(\log Q + \log n)/\gamma^2)$, and round complexity $r$.
\end{itemize}
\end{proposition}

\begin{corollary}\label{cor:noisyalgs}
There are algorithms that resample that:
\begin{itemize}
\item solve \partition\ or \select\ in the erasure model with probability $1$ with expected query complexity $O(n/\gamma)$ (based on the QuickSelect or Median-of-Medians algorithm~\cite{Hoare61,BlumFPRT73}).
\item solve \partition\ or \select\ in the erasure model with probability $1-1/\poly(n)$ with query complexity $O(n\log n/\gamma)$ and round complexity $4$.
\item solve \partition\ or \select\ in the noisy model with probability $1-1/\poly(n)$ with query complexity $O(n\log n/\gamma^2)$ and round complexity $4$.
\end{itemize} 
\end{corollary}

In the erasure model, the algorithms provided by this reduction do not have the optimal round/query complexity. We show that $\Theta(n/\gamma)$ queries are necessary and sufficient, as well as $\Theta(\log^*(n))$ rounds. For the algorithm, we begin by finding the median of a random set of size $n/\log n$ elements. This can be done in $4$ rounds and $O(n/\delta)$ total comparisons by Corollary~\ref{cor:noisyalgs}. Doing this twice in parallel, we find two elements that are guaranteed to be above/below the median, but very close. Then, we spend $\log^*(n)$ rounds comparing every element to both of these. It's not obvious that this can be done in $\log^*(n)$ rounds. Essentially what happens is that after each round, a fraction of elements are successfully compared, and we don't need to use any future comparisons on them. This lets us do even more comparisons involving the remaining elements in future rounds, so the fraction of successes actually increases with successive rounds. Analysis shows that the number of required rounds is therefore $\log^*(n)$ (instead of $\log (n)$ if the fraction was constant throughout all rounds). After this, we learn for sure that the median lies within a sublinear window, and we can again invoke the 4-round algorithm of Corollary~\ref{cor:noisyalgs} to finish up. Our lower bound essentially shows that it takes $\log^*(n)$ rounds just to have a non-erased comparison involving all $n$ elements even with $O(n/\delta)$ per round, and that this implies a lower bound. Pseudocode for the algorithm and proofs of both theorems are in Appendix~\ref{app:multinoisy}.

\begin{theorem}
\label{thm:uberasure}
With probability at least $1-1/\poly(n)$, Algorithm~\ref{alg:partitionerasure} has query complexity $O(n/\gamma)$, round complexity $\log^*(n)+8$, and solves \partition.
\end{theorem}

\begin{theorem}\label{thm:lberasure}
Assume $\gamma \leq 1/2$. In the erasure model, any algorithm solving \select\ with probability $2/3$ even with $O(n/\gamma)$ comparisons per round necessarily has round complexity $\Omega(\log^*(n))$. 
\end{theorem}

We now introduce a related problem that is strictly easier than \partition\ or \select, which we call \rank, and prove lower bounds on the round/query complexity of \rank\ noisy models, which will imply lower bounds on \partition\ and \select. In \rank, we are given as input a set $S$ of $n$ elements, and a special element $b$ and asked to determine $b$'s rank in $S$ (i.e. how many elements in $S$ are less than $b$).\footnote{A more formal statement of \rank\ is in the proof of Proposition \ref{prop:reduce} in Appendix~\ref{app:multinoisy}.} We say that a solution is a $t$-approximation if the guess is within $t$ of the element's actual rank. We show formally that \rank\ is strictly easier than \select\ in Appendix~\ref{app:technical}. From here, we prove lower bounds against \rank\ in the noisy model.

At a high level, we show (in the proof of Theorem~\ref{thm:lbnoisy}) that with only $O(n\log n/\gamma^2)$ queries, it's very likely that there are a constant fraction of $a_i$'s such that the algorithm is can't be very sure about the relation between $a_i$ and $b$. This might happen, for instance, if not many comparisons were done between $a_i$ and $b$ and they were split close to 50-50. From here, we use an anti-concentration inequality (the Berry-Essen inequality) to show that the rank of $b$ does not concentrate within some range of size $\Theta(n^{3/8})$ conditioned on the available information. In otherwords, the information available simply cannot narrow down the rank of $b$ to within a small window with decent probability, no matter how that information is used. We then conclude that no algorithms with $o(n\log n/\gamma^2)$ comparisons can approximate the rank well with probability $2/3$.

\begin{theorem}\label{thm:lbnoisy}
In the noisy model, any algorithm obtaining an $(n^{3/8}/40)$-approximation for \rank\ with probability $2/3$ necessarily has query complexity $\Omega(n\log n/\gamma^2)$.
\end{theorem}

Finally, we conclude with an algorithm for \findMin\ in the noisy model showing that \findMin\ is strictly easier than \select. This is surprising, as most existing lower bounds against \select\ are obtained by bounding \findMin. Our algorithm again begins by finding the minimum, $x$, of a random set of size $n/\log n$ using $O(n/\gamma^2)$ total comparisons by Corollary~\ref{cor:noisyalgs}. Then, we iteratively compare each element to $x$ a fixed number of times, throwing out elements that beat it too many times. Again, as we throw out elements, we get to compare the remaining elements to $x$ more and more. We're able to show that after only an appropriate number of iterations (so that only $O(n/\delta^2)$ total comparisons have been made), it's very likely that only $n/\log n$ elements remain, and that with constant probability the true minimum was not eliminated. From here, we can again invoke the algorithm of Corollary~\ref{cor:noisyalgs} to find the true minimum (assuming it wasn't eliminated). 

\begin{theorem}\label{thm:findMin}
Assume $n$ is large enough and $10 \leq c \leq \log n$. Algorithm~\ref{alg:findMin} has query complexity $\frac{3cn}{\gamma^2}$ and solves \findMin\ in the noisy model with probability at least  $1-e^{-\Omega(c)}$.
\end{theorem}

\begin{theorem}\label{thm:findMinlb}
Assume $c \geq 1$, $n$ is large enough and $\gamma \leq 1/4$. Any algorithm in the noisy model with query complexity $\frac{cn}{\gamma^2}$ solves \findMin\ with probability at most $1-e^{-O(c)}$. 
\end{theorem}

Theorem~\ref{thm:findMin} shows that \findMin\ is strictly easier than \select\ (as it can be solved with constant probability with asymptotically fewer comparisons). Theorem~\ref{thm:findMinlb} is included for completeness, and shows that it is not possible to get a better success probability without a blow-up in the query complexity. The proof of Theorem~\ref{thm:findMinlb} is similar to that of Theorem~\ref{thm:lbnoisy}.

\newpage

\bibliographystyle{alpha}
\bibliography{bib} 

\appendix


\section{Technical Lemmas}\label{app:technical}
Before beginning our proofs, we provide a few technical lemmas that will be used throughout, related to geometric sums, biased random walks, etc.

We first show the following reductions to prove that $k=n/2$ is the most difficult choice of $k$ in \select\ and \partition. We also show the reductions between \select\ and \partition.
\begin{lemma}\label{lem:relation}The following relations hold:
\begin{itemize}
\item Suppose $A$ can solve \select\ / \partition\ in the case of $n$ elements for any $n$, but only for $k = n/2$. Then $A$ can be used to solve \select\ / \partition\ for any $k,n$ with the same success probability.
\item Suppose $A$ can solve \select\ on $n$ elements. We can construct algorithm $B$ based on $A$ to solve \partition\ of $n$ elements with one more round and extra $n$ comparisons in the noiseless model, $O(n\log n/\gamma)$ in the erasure model, and $O(n\log n / \gamma^2)$ in the noisy model. The success probability decreases by $1/\poly(n)$ except in the noiseless model (where it doesn't decrease). 
\item Suppose $A$ can solve \partition\ of $n+1$ elements . We can construct algorithm $B$ based on $A$ to solve \select\ of $n$ elements with the same number of round and twice the number of comparisons. If the original success probability was $p$, the new success probability is at least $p^2$.
\end{itemize}
\end{lemma}

\begin{proof}
Let's show the reductions one by one:
\begin{itemize}
\item Wlog, let $k < n/2$. The algorithm to solve \select\ / \partition\ is the following:
\begin{enumerate}
\item Generate $n - 2k$ dummy elements which are smaller than all the $n$ elements. 
\item Run $A$ on $n$ elements together with $n-2k$ dummy elements. 
\item Output $A$'s output.
\end{enumerate}
It's easy to check the above algorithm works. 
\item $B$ is the following:
\begin{enumerate}
\item Run $A$, let $x$ be the median output.
\item In the next round, compare every element to $x$ (once in the noiseless model, $O(\log n /\gamma)$ times in the erasure model, and $O(\log n /\gamma^2)$ times in the noisy model). For each element, if it beats $x$, accept it (in the noisy model, if at least half of the comparisons beat $x$). Otherwise reject it.
\end{enumerate}
It's easy to check $B$ works.
\item $B$ is the following:
\begin{enumerate}
\item Generate two dummy elements $x_1$ and $x_2$. $x_1$ is smaller than all of the $n$ elements and $x_2$ is larger than all of the $n$ elements. 
\item Do the followings in parallel:
\begin{enumerate}
\item Run $A$ on $n$ elements and $x_1$.
\item Run $A$ on $n$ elements and $x_2$.
\end{enumerate}
\item Output the element that is accepted in the first run and rejected in the second run.
\end{enumerate}
It's easy to see that $B$ works.
\end{itemize}
\end{proof}

\begin{lemma}\label{lem:sum}
For $a \in (0,1]$, $\sum_{i=0}^\infty e^{-ai}i^k \leq 2k!/a^{k+1}$. For $a \in [1,\infty)$, $\sum_{i=0}^\infty e^{-ai}i^k \leq 2k!/a^{k}$
\end{lemma}

\begin{proof}
First, note that the function $e^{-ax}x^k$ has derivative zero exactly once in $(0,\infty)$, at $x = k/a$. So the function is increasing on $(0,k/a)$ and decreasing on $(k/a,\infty)$. This means that $\sum_{i=0}^{\lfloor k/a \rfloor -1} e^{-ai}i^k \leq \int_0^{\lfloor k/a \rfloor} e^{-ax}x^k$, and that $\sum_{\lfloor k/a \rfloor + 1}^\infty e^{-ai}i^k \leq \int_{\lfloor k/a \rfloor}^\infty e^{-ax}x^k$. Putting these together, we get:

$$\sum_{i=0}^{\infty}e^{-ai}i^k \leq \int_0^\infty e^{-ax}x^k dx + e^{-k}(k/a)^k$$

Note that $(k/e)^k \leq k!$, and $\int_0^\infty e^{-ax}x^k dx = k!/(a^{k+1})$, completing the proof.
\end{proof}

\begin{lemma}\label{lem:info}
The distribution that is heads with probability $1/2+\delta$ and tails with probability $1/2-\delta$ has $1-\Theta(\delta^2)$ bits of entropy.
\end{lemma}
\begin{proof}
\begin{eqnarray*}
&&H(1/2+\delta)\\
 &=& -( (1/2+\delta)\log_2(1/2+\delta) + (1/2-\delta)\log_2(1/2-\delta))\\
& =& 1 - (( (1/2+\delta)\log_2(1+2\delta) + (1/2-\delta)\log_2(1-2\delta)) \\
&=& 1 - \ln(2) \cdot ((1/2+\delta)(2\delta - (2\delta)^2/2) + (1/2-\delta)(-2\delta - (-2\delta)^2/2) + o(\delta^2)) ~~~~~~(\text{Taylor Expansion})\\
&=& 1 - \Theta(\delta^2).
\end{eqnarray*}
\end{proof}

\begin{proposition}\label{prop:RW}
Consider a biased random walk that moves right with probability $p \leq 1/2$ and left with probability $1-p$ at every step. Then the probability that this random walk reaches $k$ units right of the origin at any point in time is exactly $(p/(1-p))^{k}$.
\end{proposition}
\begin{proof}
First, note that the probability that this random walk reaches $k$ units right of the origin at any point in time is exactly the probability that a random walk with the same bias reaches $1$ unit right of the origin $k$ times independently, because once the random walk reaches $1$ unit right of the origin, the remaining random walk acts like a fresh random walk that now only needs to move $k-1$ units to the right at some point in time. So we just need to show that the probability that the random walk moves $1$ unit to the right at some point in time is $p/(1-p)$. 

Note that whatever this probability, $q$, is, it satisfies the equality $q = p + (1-p)q^2$. This is because the probability that the random walk moves right is equal to the probability that the random walk moves right on its first step, plus the probability that the random walk moves left on its first step, and then moves two units right at some point in time. This equation has two solutions, $q = p/(1-p)$ and $q = 1$. So now we just need to show that $q \neq 1$ when $p < 1/2$.

Assume for contradiction that $q = 1$. Then this means that the random walk not only reaches one unit right of the origin once during the course of the random walk, but that it reaches one unit right of the origin infinitely many times, as every time the walk reaches the origin is a fresh random walk that moves one unit right with probability one. So let $A_j$ denote the random variable that is $1$ if the walk moves right at time $j$, and $-1$ if the walk moves left at time $j$. We have just argued that if $q=1$, there are infinitely many $t$ such that $\sum_{j \leq t} A_j/t > 0$. Therefore, $\lim \inf_{t \rightarrow \infty} \sum_{j \leq t} A_j/t \geq 0$. However, we also know that $\mathbb{E}[\sum_{j \leq t} A_j/t] = 1-2p < 0$, and the $A_j$s are independent. So the law of large numbers states that $\lim_{t \rightarrow \infty} \sum_{j \leq t} A_j/t = 1-2p$, a contradiction.

\end{proof}

We also include a technical lemma confirming that it is okay to do all of our sampling without replacement, if desired.

\begin{lemma}
\label{lem:rep}
Let $S$ be any set of $n$ elements, all in $[0,1]$. Let $X_1,\ldots,X_k$ be $k$ samples from $S$ without replacement. Then $Pr[|\sum_i X_i - (k/n)\sum_{s \in S} s | \geq \delta k] \leq 2e^{-\delta^2 k/2}$.
\end{lemma}
\begin{proof}
Define $Y_j = \mathbb{E}[\sum_i X_i | X_1,\ldots,X_j]$. It's clear that the $Y_j$, $j = 1, \ldots,n$ form a martingale (specifically, the Doob martingale for $\sum_i X_i$), and that $Y_0 = \mathbb{E}[\sum_i X_i] = (k/n)\sum_j y_j$. So we just need to reason about how much the conditional expectation can possibly change upon learning a single $X_i$, then we can apply Azuma's inequality.

Conditioned on $X_1,\ldots,X_j$, each of the remaining $n-j$ elements of $S$ are equally likely to be chosen, and each is chosen with probability exactly $(k-j)/(n-j)$. So $Y_j = \sum_{i=1}^j X_i + \sum_{\text{unsampled } s \in S} \frac{k-j}{n-j}s$. How much can $|Y_{j}-Y_{j-1}|$ possibly be? We have (below, unsampled means elements that are still unsampled even after step $j$):
$$Y_{j}-Y_{j-1} = \frac{n-k}{n-j}X_{j} + \sum_{\text{unsampled } s \in S} (\frac{k-j}{n-j} - \frac{k-j+1}{n-j+1})s$$
$$ = \frac{n-k}{n-j}X_j - \frac{n-k}{(n-j)(n-j+1)} \sum_{\text{unsampled } s \in S} s$$

It is clear that the above quantity is at most $1$, because all $s \in [0,1]$, and $\frac{n-k}{n-j} \in [0,1]$ as well. It is also clear that the above quantity is at least $-1$, as $\sum_{\text{unsampled } s \in S} s \leq n-j$, and $\frac{n-k}{n-j+1} \leq 1$. So the Doob martingale has differences at most $1$ at each step, and a direct application of Azuma's inequality yields the desired bound.
\end{proof}

\section{Proofs for Non-Adaptive Algorithms}
\label{app:oneround}

\subsection{Upper Bounds}
\begin{algorithm}[ht]
        \caption{Non-adaptive procedure for querying $dn$ comparisons}
    \begin{algorithmic}[1]\label{alg:skeleton}
        \STATE Select a \emph{skeleton set}, $S$, of size $\sqrt{n}$ (without replacement). 
        \STATE Compare each element of $S$ to each other element of $S$.
        \STATE Compare each element not in $S$ to $d-1$ random elements of $S$.
             \end{algorithmic}
\end{algorithm}

\subsubsection{Noiseless Model}

\begin{algorithm}[ht]
        \caption{Non-adaptive algorithm for the noiseless model}
    \begin{algorithmic}[1]\label{alg:nanoiseless}
        \STATE Run Algorithm~\ref{alg:skeleton}. Let $S$ denote the skeleton set selected, and $x$ denote the median of $S$. 
        \STATE Denote by $A_S$ the subset of $S$ that beat $x$. Denote by $R_S = S-A_S$. 
        \STATE For all $i \in S$, accept $i$ iff $i \in A_S$. Otherwise, reject. 
	\STATE For all $i \notin S$, if $i$ beats an element in $A_S$, accept. If an element in $R_S$ beats $i$, reject. Otherwise, make a random decision for $i$.
             \end{algorithmic}
\end{algorithm}

\begin{prevproof}{Theorem}{thm:nanoiseless}
We consider the error contributed by $n/2+i$ (which is the same as $n/2-i$). There are two events that might cause $n/2+i$ to be misplaced. First, maybe $n/2+i$ loses to $y$ for some $y \in S,y < x$. This is unlikely because this can only happen in the event that $x$ is a poor approximation to the median. Second, maybe $n/2 + i$ is never compared to an element $y \in [x,n/2+i]$. This is also unlikely because the fraction of such elements should be about $i/n$. We bound the probabilty of the former event first.

\begin{lemma}\label{lem:skeleton}
Let $X_i$ denote the fraction of elements in $S$ greater than $n/2+i$. Then $X_i \leq 1/2 - i/n + \varepsilon$ with probability at least $1-2e^{-\varepsilon^2\sqrt{n}/2} $.
\end{lemma}
\begin{proof}
$X_i$ is the average of $\sqrt{n}$ independent random variables, each denoting whether $n/2+i <$ some $y\in S$. $E[X_i] = 1/2 -i/n$. Applying Lemma \ref{lem:rep} yields the lemma.
\end{proof}

We call a skeleton set $S$ is good if for $i = 1,..,n/2$, $X_i \leq 1/2 - i/n + \varepsilon$. Now let's fix the skeleton set $S$ and assume $S$ is good. From the above lemma, we know $S$ is good with probability at least $1-n\cdot e^{-\varepsilon^2\sqrt{n}/2}$. 

\begin{lemma}
If $S$ is good, then the probability that $n/2+i$ is rejected is at most $(1-i/n+\varepsilon)^{d-1}$ (and at most $1$ if $i/n \leq \varepsilon$).
\end{lemma}

\begin{proof}
The elements that $n/2+i$ are compared to are chosen uniformly at random from $S$. At least a $1/2 +i/n - \varepsilon$ fraction of them are less than $n/2+i$, so an $i/n - \varepsilon$ fraction of elements in $S$ lie in $[x,n/2+i]$. So each time we choose a random element of $S$ to compare to $n/2+i$, we have at least a $\max\{i/n-\varepsilon,0\}$ chance of comparing to an element in $[x,n/2+i]$. In the event that this happens, we are guaranteed to accept $n/2+i$. The probability that we miss on each of $d-1$ independent trials is exactly $(1-i/n+\varepsilon)^{d-1}$. 
\end{proof}

Therefore, conditioning that $S$ is good, the expected $c$-weighted error contributed by elements $n/2+1,...,n$ is at most : 

$$\sum_{i=1}^{n/2} \min\{1,(1-i/n+\varepsilon)\}^{d-1}i^c \leq \sum_{i=1}^{\varepsilon n} i^c + \sum_{i=\varepsilon n}^{n/2} (1-i/n+\varepsilon)^{d-1}i^c$$
$$ \sum_{i=\varepsilon n}^{n/2} (1-i/n+\varepsilon)^{d-1}i^c \leq \sum_{i=\varepsilon n}^{n/2} e^{(d-1)(-i/n+\varepsilon)}i^c$$
$$ \leq e^{d\varepsilon} \sum_{i=1}^{n/2} e^{-(d-1)i/n}i^c$$
$$\leq 2e^{d\varepsilon} c! (n/(d-1))^{c+1}$$

The last inequality is a corollary of Lemma~\ref{lem:sum} proved in Appendix~\ref{app:technical}.  Taking $\varepsilon = 1/d$ shows that the expected $c$-weighted error contributed by elements $n/2+1,...,n$ is $\frac{(n\varepsilon)^{c+1}}{c+1} + \frac{2n^{c+1}e^{d\varepsilon}c!}{(d-1)^{c+1}} = O((n/d)^{c+1})$ conditioned on $S$ is good. Notice that when $S$ is fixed, the $c$-weighted error contributed by each element is independent and bounded by $n^c$. By the Hoeffding bound, the probability that the $c$-weighted error exceeds its expectation by more than $(n/d)^{c+1}$ is at most $e^{-2n^3/d^{2c+2}}$ conditioned on $S$ is good. To sum up, taking elements that are smaller than the median into account, we know that with probability at least $1-2n\cdot e^{-\sqrt{n}/(2d^2)} - 2e^{-2n^3/d^{2c+2}}$, the $c$-weighted error is $O((n/d)^{c+1})$.
\end{prevproof}

Here we also show that our algorithm is better than a simpler solution. The simpler solution would be to just compare every element to a random $d$ elements, and accept if it is larger than at least half, and reject otherwise. We show that, unlike the algorithm above, this doesn't obtain asymptotically optimal error as $d$ grows. 

\begin{theorem}
The simple solution has expected error $\Omega(n^{c+1}/d^{(c+1)/2})$. 
\end{theorem}
\begin{proof}
Let $X_{ij}$ be an indicator variable for the event that $n/2+i$ is smaller than the $j^{th}$ element it is compared to. Then $n/2+i$ is accepted iff $\sum_{j} X_{ij} > d/2$. As each $X_{ij}$ is a bernoulli random variable that is $1$ with probability $1/2 - i/n$, the probability that $n/2+i$ is mistakenly rejected is exactly the probability that a $B(d,1/2-i/n)$ random variable exceeds its expectation by at least $di/n$. This happens with probability on the order of $e^{-(i/n)^2d}$. So for $i = O(n/\sqrt{d})$,  this is at least $1/e$, meaning that the error contribution from all $n/2+i$, $i \leq k n/\sqrt{d}$ for some absolute constant $k$ is at least $i^c/e$. Summing over all $i$ means that the total error is at least $\sum_{i = 1}^{kn/\sqrt{d}} i^c/e = \Omega (n^{c+1}/d^{(c+1)/2})$.
\end{proof}

\subsubsection{Erasure Model}

\begin{algorithm}[ht]
        \caption{Non-adaptive algorithm for the erasure model}
    \begin{algorithmic}[1]\label{alg:naerasure}
        \STATE Run Algorithm~\ref{alg:skeleton}. Let $S$ denote the skeleton set selected.
	\STATE Say that element $a \in S$ is \emph{known to beat} element $b \in S$ if there exists some $a = s_0 > s_1 > \ldots > s_\ell = b$ with all $s_j \in S$, and $s_j$ beats $s_{j+1}$ for all $j \in \{0,\ldots,\ell-1\}$ (i.e. all of these comparisons were not erased).
        \STATE Denote by $A_S$ the elements of $S$ that are known to beat at least $|S|/2$ elements of $S$. Denote by $R_S$ the elements of $S$ that are known to be beaten by at least $|S|/2$ elements of $S$ (note that $S$ may not equal $A_S \cup R_S$). 
        \STATE For all $i \in S$, accept $i$ if $i \in A_S$. Reject $i$ if $i \in R_S$. Otherwise, make a random decision for $i$.
	\STATE For all $i \notin S$, if $i$ beats an element in $A_S$, accept. If an element in $R_S$ beats $i$, reject. Otherwise, make a random decision for $i$.
             \end{algorithmic}
\end{algorithm}

\begin{prevproof}{Theorem}{thm:naerasure}
We again consider the error contributed by $n/2+i$ (which is the same as $n/2-i$). There are again two events that might cause $n/2+i$ to be misplaced. First, maybe it is beaten an element in $R_S$. This is unlikely because this can only happen in the event that some element below the median makes it into $A_S$. Second, maybe $i$ is never compared to an element that beats it in $A_S$. This is also unlikely because the fraction of such elements should be about $i/n$. We bound the probabilty of the former event first, making use of Lemma~\ref{lem:skeleton}.

Again let $X_i$ denote the fraction of elements in $S$ that are smaller than $n/2+i$, and define $S$ to be good if $X_i \leq 1/2 - i/n + \varepsilon$ for all $i \in [1,n/2]$. Then Lemma~\ref{lem:skeleton} guarantees that $S$ is good with probability at least $1-ne^{-\varepsilon^2\sqrt{n}/2}$. This time, in addition to $S$ being good, we also need to make sure that $A_S$ is large (which will happen as long as not too many comparisons are erased).

\begin{lemma}\label{lem:pairs}
Let $x, y \in S$ such that $x>y$ and $|S \cap (x,y)| = k$. Then with probability at least $1-e^{-k\gamma^2}$, $x$ is known to beat $y$.
\end{lemma}
\begin{proof}
$x$ is known to beat $y$ if there is some $z \in S \cap (x,y)$ such that the comparisons between $x$ and $z$ and $y$ and $z$ are both not erased. There are $k$ such possible $z$, and all comparisons are erased independently. So the probability that for all $z$, at least one of the comparisons to $\{x,y\}$ were erased is $(1-\gamma^2)^k \leq e^{-k\gamma^2}$. 
\end{proof}

\begin{corollary}
For all $\varepsilon \in (0,1/2)$, with probability at least $1-n^2e^{-\varepsilon\gamma^2\sqrt{n}}$, both $A_S$ and $R_S$ have at least $(1/2 - \varepsilon)|S|$ elements.
\end{corollary}
\begin{proof}
By Lemma~\ref{lem:pairs} and a union bound, with probability at least $1-n^2e^{-\varepsilon\gamma^2\sqrt{n}}$, it is the case that for all $x, y \in S$ that have at least $\varepsilon |S|$ elements between them, it is known whether $x$ beats $y$ or vice versa. In the event that this happens, any element that is at least $\varepsilon |S|$ elements away from the median will be in $A_S$ or $R_S$. 
\end{proof}

We'll call a skeleton set $S$ really good if it is good, and $|A_S| \geq (1/2-\varepsilon)|S|$. Now let's fix the skeleton set $S$ and assume $S$ is really good. From the above arguments, we know $S$ is really good with probability at least $1-n\cdot e^{-\varepsilon^2\sqrt{n}/2}-n^2e^{-\varepsilon\gamma^2\sqrt{n}}$. 

Next, observe that if $X_i \leq 1/2 - i/n + \varepsilon$, and $|A_S| \geq (1/2 - \varepsilon)|S|$, then there are at least $(i/n - 2\varepsilon)|S|$ elements in $A_S$ less than $n/2+i$. Therefore the probability that $n/2+i$ never beats an element in $A_S$ is at most $(1-(i/n - 2\varepsilon)\gamma)^{d-1}$ (and at most $1$ if $i/n \leq 2\varepsilon$). So the total $c$-weighted error that comes from these cases is at most $\min\{1,(1-(i/n-2\varepsilon)\gamma)^{d-1}\}i^c$.

Conditioning on $S$ is really good, the expected $c$-weighted error is 

$$\sum_{i=1}^{n/2} \min\{1,(1-(i/n-2\varepsilon)\gamma)\}^{d-1}i^c \leq \sum_{i=1}^{2\varepsilon n} i^c + \sum_{i=2\varepsilon n}^{n/2} (1-(i/n-2\varepsilon)\gamma)^{d-1}i^c$$
$$ \sum_{i=2\varepsilon n}^{n/2} (1-(i/n-2\varepsilon)\gamma)^{d-1}i^c \leq \sum_{i=2\varepsilon n}^{n/2} e^{(d-1)\gamma(-i/n+2\varepsilon)}i^c$$
$$ \leq e^{2d\varepsilon\gamma} \sum_{i=1}^{n/2} e^{-(d-1)\gamma i/n}i^c$$
$$ \leq \frac{e^{2d\varepsilon\gamma}c!n^{c+1}}{(\gamma(d-1))^{c+1}}$$

Taking $\varepsilon = 1/(\gamma d)$ shows that the expected $c$-weighted error contributed by elements $n/2+1,...,n$ is $n^{c+1}\left( \frac{(2\varepsilon)^{c+1}}{c+1} + \frac{e^{2d\varepsilon\gamma}c!}{(\gamma(d-1))^{c+1}}\right) = O((n/(\gamma d))^{c+1})$ conditioned on $S$ is really good. Notice that when $S$ is fixed, the $c$-weighted error contributed by each element is independent and bounded by $n^c$. By the Hoeffding bound, the probability that the $c$-weighted error exceeds its expectation by more than $(n/(\gamma d))^{c+1}$ is at most $e^{-2n^3/(\gamma d)^{2c+2}}$, conditioned on $S$ is really good. To sum up, taking elements that are smaller than the median into account, we know that with probability at least $1-2n\cdot e^{-\sqrt{n}/(2\gamma^2d^2)}-2n^2e^{-\gamma\sqrt{n}/d} - 2e^{-2n^3/(\gamma d)^{2c+2}}$, the $c$-weighted error is $O((n/(\gamma d))^{c+1})$.

\end{prevproof}

\subsubsection{Noisy Model}

\begin{algorithm}[ht]
        \caption{Non-adaptive algorithm for the noisy model}
    \begin{algorithmic}[1]\label{alg:nanoisy}
        \STATE Run Algorithm~\ref{alg:skeleton}. Let $S$ denote the skeleton set selected.
	\STATE Run the Braverman-Mossel algorithm to recover the maximum likelihood ordering of elements in $S$ in time $\poly(n,1/\gamma)$~\cite{BravermanM08}. Let $x$ denote the median under this ordering, and denote by $y <_S z$ that $y$ comes before $z$ in this ordering.
        \STATE For all $i \in S$, accept $i$ iff $i >_S x$, otherwise reject.
	\STATE For all $i \notin S$, let $b_{i1} <_S \ldots <_S b_{ic}$ denote the elements $>_S x$ that $i$ is compared to (note that maybe $c = 0$), and $\ell_{i1} >_S \ldots >_S \ell_{if}$ denote the elements $<_S x$ that $i$ is compared to. 
	\STATE Let $X_{ij} = 1$ iff $i$ beats $b_{ij}$, and $-1$ otherwise. Let $Y_{ij} = 1$ iff $i$ beats $\ell_{ij}$ and $-1$ otherwise. Let $B_i = \max\{J |\sum_{j=1}^J X_{ij} \geq 0\}$ and $L_i = \max\{J |\sum_{j=1}^J Y_{ij} \leq 0\}$. If $B_i > L_i$, accept $i$. If $B_i < L_i$, reject. If $B_i = L_i$, make a random decision.
             \end{algorithmic}
\end{algorithm}

\begin{prevproof}{Theorem}{thm:nanoisy}
Again we consider the probability of mistakenly accepting the element $n/2+i$. This can happen for the same reasons as in the previous sections: perhaps the elements selected for $S$ are not representative of the ground set, or $n/2+i$ is not compared to elements that separate it from the median. Additionally, mistakes may now happen due to erroneous comparisons, even when $n/2+i$ is compared to ``the right'' elements. We proceed now to bound the probability of each bad event, beginning with the event that $S$ is poorly sorted due to comparison errors.

\begin{theorem}\label{thm:BM}(\cite{BravermanM08}) For each element $e \in S$, let $\pi(e)$ denote the true rank of $e$ when $S$ is properly sorted (i.e. the best element in $S$ has rank $1$), and let $\sigma(e)$ denote the rank of $e$ in the maximum likelihood ordering after taking all $\binom{|S|}{2}$ pairwise comparisons. Then there is an absolute constant $K$ such that for all $\varepsilon$ and for all $e$ simultaneously, with probability at least $1-2\varepsilon$, $|\pi(e)-\sigma(e)| \leq K\gamma^{-4}(\log |S| - \log \varepsilon)$. In particular, taking $\varepsilon = 1/\poly(n)$ and recalling that $|S| = \sqrt{n}$ implies that with probability at least $1-2/\poly(n)$, $|\pi(e)-\sigma(e)| = O(\log(n)/\gamma^4) = o(|S|)$. 
\end{theorem}

Now, we want to combine this with Lemma~\ref{lem:skeleton} to claim that there are not only many elements in $S$ smaller than $n/2+i$, but that these elements are sorted well enough in $S$ for this to be useful. Recall that Lemma~\ref{lem:skeleton} states that for all $\varepsilon$, with probability at least $1-e^{-\varepsilon^2\sqrt{n}/2}$, the fraction of elements in $S$ larger than $n/2+i$ is at most $1/2-i/n+\varepsilon$. Combining Lemma~\ref{lem:skeleton} with Theorem~\ref{thm:BM} yields the following:

\begin{corollary}\label{cor:BM}
For all $\varepsilon$, with probability at least $1-2/\poly(n)-e^{-\varepsilon^2\sqrt{n}/2}$, there is a set $T\subset S$ such that:
\begin{itemize}
\item $|T| \geq (i/n-\varepsilon-o(1))|S|$.
\item Every element in $T$ is $>_S x$.
\item Every element in $T$ is smaller than $n/2+i$.
\item Every element in $T$ is $<_S y$, for all $y \in S \cap (n/2+i,n]$.
\end{itemize}
\end{corollary}
\begin{proof}
By Lemma~\ref{lem:skeleton}, there are at least $(1/2+i/n-\varepsilon)|S|$ elements in $S$ smaller than $n/2+i$ (except with probability $e^{-\varepsilon^2\sqrt{n}/2}$). Therefore, by Theorem~\ref{thm:BM}, the maximum likelihood ordering cannot place any element $ > n/2+i$ in any rank below $(1/2+i/n-\varepsilon - o(1))|S|$ (except with probability $1/n^2$. So simply define $T$ to be the set of elements ranked between $|S|/2$ and $(1/2 + i/n-\varepsilon - o(1))|S|$, and $T$ necessarily satisfies the desired properties (except with the stated probability).
\end{proof}

From here, the idea is that $B_{n/2+i}$ is likely to be large, as the elements in $T$ form a large set immediately above $x$ such that $n/2+i$ should beat more than half the elements in whatever subset of $T$ it is compared to. Similarly, we expect $L_{n/2+i}$ to be small, because $n/2+i$ shouldn't lose to a large set of elements below $x$. We first show that $n/2+i$ is likely to be compared to many elements of $T$.

\begin{lemma}\label{lem:numcompares}
Let $T$ be any set of size $\delta|S|$ elements. Then with probability at least $1-e^{-\delta (d-1)/8}$, $n/2+i$ is compared to at least $\delta (d-1)/2$ elements of $T$. 
\end{lemma} 

\begin{proof}
$n/2+i$ is compared to a total of $d-1$ elements of $S$, chosen independently and uniformly at random. So each comparison is to an element in $T$ with probability $\delta$, and the expected number of elements in $T$ that $n/2+i$ is compared to is therefore $\delta (d-1)$. Let $A_j$ be the indicator random variable that is $1$ if the $j^{th}$ comparison is to an element in $T$ and $0$ otherwise. Then the number of elements in $T$ that $n/2+i$ is compared to is $\sum_j A_j$ and is the sum of independent $\{0,1\}$ random variables. Therefore, the Chernoff bound states that $Pr[\sum_j A_j \leq \delta (d-1)/2] \leq e^{-\delta (d-1)/8}$. 
\end{proof}

\begin{lemma}\label{lem:qualitycompares}
Let $T\subseteq S$ be any set of elements such that:
\begin{itemize}
\item Every element in $T$ is $>_S x$.
\item Every element in $T$ is smaller than $n/2+i$.
\item Every element in $T$ is $<_S y$, for all $y \in S \cap (n/2+i,n]$.
\end{itemize}
Then if $n/2+i$ is compared to at least $t$ elements of $T$, $Pr[B_{n/2+i} \geq t] \geq 1-e^{-t\gamma^2/8}$. 
\end{lemma}
\begin{proof}
Let $T'$ denote the subset of $T$ that $n/2+i$ is compared to. Define $A_j$ to be a random variable that is $1$ if $n/2+i$ beats the $j^{th}$ element of $T'$, and $0$ otherwise. Then $B_{n/2+i}\geq |T'|$ if  $\sum_{j=1}^{|T'|} A_j \geq |T'|/2$. As the $A_j$s are independent $\{0,1\}$ random variables that are $1$ with probability $1/2+\gamma/2$, the Chernoff bound states that this probability is at least $1-e^{-|T'|\gamma^2/8}$. 
\end{proof}

Finally, we show that $L_{n/2+i}$ is likely to be small, and put everything together. Recall that $\ell_1,\ldots,\ell_c$ denote the elements that $n/2+i$ are compared to that are $<_S x$ in decreasing order, and that $Y_{ij} = 1$ iff $n/2+i$ beats $\ell_j$, and $-1$ otherwise. 

\begin{lemma}\label{lem:startlow}
Assume that every element $<_S x$ is also $< n/2+i$. Then for all $t$, $Pr[\sum_{j \leq t} Y_{ij} \geq t\gamma/2] \geq 1-e^{-t\gamma^2/32}$. 
\end{lemma}

\begin{proof}
The random variables $(1+Y_{ij})/2$ are independent and in $\{0,1\}$. $\mathbb{E}[\sum_{j\leq t} Y_{ij}] = (1/2+\gamma/2)t$. So the Chernoff bound states that $Pr[\sum_{j\leq t}(1+Y_{ij})/2 \leq (1/2 +\gamma/4)t] \leq e^{-t\gamma^2/32}$. Note that $\sum_{j \leq t}(1+Y_{ij})/2 \leq (1/2 + \gamma/4)t \Leftrightarrow \sum_{j \leq t} Y_{ij} \leq t\gamma/2$. 
\end{proof}

\begin{proposition}\label{prop:finishlow}
Assume that every element $<_S x$ is also $< n/2+i$. Then for all $t$, $Pr[L_{n/2+i} \geq t] \leq e^{-t\gamma^2/32} + e^{-t\gamma^2/2}$. 
\end{proposition}
\begin{proof}
Define a random walk starting at $0$ that moves left at time $j$ if $Y_{ij} = -1$, and right at time $j$ if $Y_{ij} = 1$. Then the event that $L_{n/2+i} \geq t$ is exactly the event that this random walk returns to the origin (or left of the origin) at some step $\geq t$. If $\sum_{j\leq t} Y_{ij} \geq t\gamma/2$, then this is exactly the event that the random walk starting at time $t$ winds up $t\gamma/2$ steps to the left. Note that this random walk moves right with probability $1/2+\gamma/2$ and left with probability $1/2-\gamma/2$, so the probability that this event occurs is exactly $(\frac{1/2-\gamma/2}{1/2+\gamma/2})^{t\gamma/2}$ by Proposition~\ref{prop:RW}. 

Finally, we observe:
$$\left(\frac{1/2-\gamma/2}{1/2+\gamma/2}\right)^{t\gamma/2} \leq \left(1-\gamma\right)^{t\gamma/2} \leq e^{-t\gamma^2/2}$$
\end{proof}

Now, we are ready to put everything together to prove the theorem. By Corollary~\ref{cor:BM}, with probability at least $1-2/\poly(n)-e^{-\varepsilon^2\sqrt{n}/2}$, there is a set of size $(i/n-\varepsilon-o(1))|S|$ of elements immediately $>_S x$ that are all $< n/2+i$. By Lemmas~\ref{lem:numcompares} and~\ref{lem:qualitycompares}, in the event that this happens, $B_{n/2+i} \geq (d-1)(i/n-\varepsilon - o(1))/2$ with probability at least $1-e^{-(i/n-\varepsilon-o(1))(d-1)/8}-e^{-(i/n-\varepsilon-o(1))(d-1)\gamma^2/2}$. Furthermore, by Proposition~\ref{prop:finishlow}, in the event that this happens, $L_{n/2+i} < (d-1)(i/n-\varepsilon - o(1))/2$ with probability at least $1-e^{-(i/n-\varepsilon-o(1))(d-1)\gamma^2/64}-e^{-(i/n-\varepsilon-o(1))(d-1)\gamma^2/4}$. 

Combining all of this with a union bound, we see that $n/2+i$ is mistakenly accepted with probability at most $2/\poly(n)+e^{-\varepsilon^2\sqrt{n}/2}+4e^{-(i/n-\varepsilon)(d-1)\gamma^2/64}$. Taking $\varepsilon = 1/d$ and summing this over all $i$, we get:

$$\sum_{i=1}^{n/2} (2/\poly(n)+e^{-\sqrt{n}/2d^2}) i^c = o(n^{c+1}).$$
$$\sum_{i=1}^{n/2} \left(4e^{-(i/n-\varepsilon)(d-1)\gamma^2/64}\right)i^c$$
$$\leq 4e^{\gamma^2/64}\sum_{i=1}^{n/2}e^{-((d-1)\gamma^2/64n)i}i^c$$
$$\leq 8c!\left(64nd^{-1}\gamma^{-2}\right)^{c+1}.$$
The last inequality is due to Lemma~\ref{lem:sum}. Summing both terms together yields the theorem.

\end{prevproof}

\subsection{Lower Bounds}
In the three subsections below, we provide proofs of the three lower bounds contained in Theorem~\ref{thm:nalowerbounds}. 

\subsubsection{Noiseless Model}
The main idea of the proof is that if two elements $i < n/2$ and $j > n/2$ are never compared to any elements in $[i,j]$, then the results of these comparisons are independent of whether or not $i$ and $j$ are swapped. So however the algorithm decides to place $i$ and $j$, a coupling argument (by swapping the role of $i$ and $j$ in all comparisons) shows that the error must be large if $i$ and $j$ are unlikely to be compared to elements in $[i,j]$. We proceed by analyzing pairs of elements each a distance $i$ away from $n/2$ separately. 

Let $D$ be any distribution over set of comparisons chosen non-adaptively for the algorithm. As the elements are a priori indistinguishable, it's clear that $D$ is invariant under permutations. So let $S$ denote a sample from $D$, and $S_i$ denote the same set after swapping the roles of $n/2 -i$ and $n/2 + i$ in all comparisons. Let also $X_i(S)$ denote the random variable that equal to the number of $\{n/2+i,n/2-i\}$ that are incorrectly classified by the algorithm with comparisons $S$. Then clearly the expected error is equal to $\sum_{S} \sum_{i = 1}^{n/2} Pr[S] E[X_i(S)]i^c$. And because $D$ is invariant under permutations, we can rewrite this as:

$$\sum_S \sum_{i =1}^{n/2} Pr[S] \frac{E[X_i(S)] + E[X_i(S_i)]}{2} i^c$$

Now we proceed to bound $E[X_i(S)] + E[X_i(S_i)]$.

\begin{lemma}
Let $p_i$ denote the probability that neither of $\{n/2+i,n/2-i\}$ are compared to an element in $[n/2-i, n/2+i]$ when sampling a set of comparisons from $D$. Then $\sum_S\frac{E[X_i(S)] + E[X_i(S_i)]}{2} \geq p_i$.
\end{lemma}
\begin{proof}
Because neither element is compared to an element between them, the result of all comparisons will be the same using $S$ or $S_i$. So coupling the randomness used by any algorithm with $S$ versus $S_i$, we see that swapping the elements $n/2+i, n/2-i$ simply swaps the set that each element winds up in. If both are placed in $R$ using $S$, then both are placed in $R$ using $S_i$, so each produces one mistake. If both are placed in $A$ using $S$, then both are placed in $A$ using $S_i$, so each produces one mistake. If one is placed in $A$ and the other in $R$ using $S$, then the same holds using $S_i$, but swapped. So either $S$ produces zero mistakes and $S_i$ produces two, or vice versa. So no matter the algorithm, it is clear that when $S$ does not contain comparisons between $\{n/2+i, n/2-i\}$ and any elements between them, that $E[X_i(S)] + E[X_i(S_i)] = 2$. 
\end{proof}

\begin{corollary}
Let $p_i$ denote the probability that neither of $\{n/2+i, n/2-i\}$ are compared to an element in $[n/2-i, n/2+i]$ during algorithm $\mathcal{A}$. Then the expected error of $\mathcal{A}$ is at least $\sum_{i =1}^{n/2} p_i i^c$. 
\end{corollary}

The remaining work is to lower bound $p_i$ and evaluate the sum.

\begin{lemma}
For any algorithm $\mathcal{A}$ that makes at most $dn$ non-adaptive queries, $p_i \geq \max\{ 1 - 8di/n,0\}$.
\end{lemma}
\begin{proof}
Because $D$ is invariant under permutations, we can sample $S$ by first building a graph with $n$ nodes and edges between two nodes that are compared, then randomly labeling the nodes with $1,\ldots, n$. For any pair $(x, y)$, the probability that they are compared is then the probability that two randomly selected nodes have an edge between them. If there are $dn$ total edges in the graph, this probability is (almost) exactly $2d/n$. Taking a union bound over the $4i$ possible comparisons means that the probability that either $n/2-i $ or $n/2+i$ is compared to something in between is at most $8di/n$. 
\end{proof}

Now, we may complete the proof by evaluating the sum.

$$\sum_{i=1}^{n/2} p_i i^c \geq \sum_{i=1}^{n/8d} p_i i^c$$
$$ \geq \sum_{i = 1}^{n/8d} (1-8di/n)i^c$$
$$ \geq \sum_{i=1}^{n/8d} i^c - \sum_{i=1}^{n/8d} 8di^{c+1}/n$$
$$ \geq (n/8d)^{c+1}/(c+1) - 8d(n/8d)^{c+2}/n(c+2)$$
$$ = (n/8d)^{c+1}/(c+1)(c+2) = \Omega((n/d)^{c+1})$$.

\subsubsection{Erasure Model}
We again consider swapping the role of $n/2+i$ and $n/2-i$ in all comparisons. This time, we also have to couple whether or not each comparison is erased in addition to which comparisons are made. Recall that each comparison is erased with probability $1 - \gamma$ and correct with probability $\gamma$. Now, instead of just sampling which elements are compared, also sample whether the comparison will be correct or erased.

It is clear that any erased comparison will give the same output whether the comparison is made to $n/2+i$ or $n/2-i$, and any comparison (correct or erased) made from an element in $\{n/2-i,n/2+i\}$ to an element outside of $[n/2-i,n/2+i]$ will also give the same output for either comparison. So the probability that the outcome of any comparison is different for $n/2+i$ and $n/2-i$ is exactly the probability that there is a correct comparison from an element in $\{n/2-i,n/2+i\}$ to an element in $[n/2-i, n/2+i]$. 

So again, there are $4i$ possible comparisons that compare $n/2-i$ or $n/2+i$ to an element between them. For each comparison, there is probability $\gamma$ that the comparison is correct, and probability $2d/n$ that it is made. So the probabilty that any fixed comparison is made and correct is $2\gamma d/n$. Taking a union bound over all $4i$ possible comparisons means that the probability that none of these comparisons are made is at least $1-8\gamma d/n$. From here, we may take exactly the same sum as in the noiseless case, multiplying the appropriate terms by $\gamma$, and obtain the theorem:

$$\text{Error} \geq \sum_{i = 1}^{n/8d\gamma} (1-8d\gamma i/n)i^c$$
$$ \geq \sum_{i=1}^{n/8d\gamma} i^c - \sum_{i=1}^{n/8d\gamma} 8d\gamma i^{c+1}/n$$
$$ \geq (n/8d\gamma)^{c+1}/(c+1) - 8d\gamma(n/8d\gamma)^{c+2}/n(c+2)$$
$$ = (n/8d\gamma)^{c+1}/(c+1)(c+2) = \Omega( (n/(d\gamma))^{c+1})$$.
\subsubsection{Noisy Model}
 We again consider the algorithms performance on the pair $\{n/2+i,n/2-i\}$, but this time we will have to reason using information theory instead of first principles. We will describe a reduction that shows that any algorithm that makes few errors on this pair necessarily learns a lot of information about their relative locations, and that this can only happen if the algorithm makes many queries between $\{n/2+i,n/2-i\}$ and $[n/2-i,n/2+i]$. 

Consider the following problem, \textsc{Which-is-Which}:\\
\textsc{Input}: A random permutation of the elements $[n]$. All elements are labeled except for $n/2+i$ and $n/2-i$. The input may be accessed via comparisons between two elements (labeled or unlabeled), these comparisons will be correct with probability $\gamma$, and random with probability $1-\gamma$. \\
\textsc{Output}: A label for each of the unlabeled elements from $\{n/2+i,n/2-i\}$. \\
\textsc{Goal}: Maximize the probability that the elements are labeled correctly. 

\begin{lemma}
Any algorithm that answers \textsc{Which-is-Which} correctly with probability $1/2+\delta$ learns $\Theta(\delta^2)$ bits of information (in expectation).
\end{lemma}

\begin{proof}
A solver running such an algorithm for \textsc{Which-is-Which} begins with a prior distribution that each label is equally likely, which has one bit of entropy. After running the algorithm, the solver's posterior distribution has entropy $1-\Theta(\delta^2)$, by Lemma~\ref{lem:info}, which is proved in Appendix~\ref{app:technical}. Therefore, the solver must have learned $\Theta(\delta^2)$ bits of information to update her prior to her posterior.
\end{proof}

\begin{lemma}
Any algorithm that learns $b$ bits of information in expectation necessarily makes $\Omega(b/\gamma^2)$ comparisons between elements in $\{n/2+i,n/2-i\}$ and $[n/2-i,n/2+i]$ in expectation. 
\end{lemma}
\begin{proof}
Each comparison involving two labeled elements provides no new information, as these elements are already labeled. Any comparison between unlabeled elements (elements in $\{n/2+i,n/2-i\}$) and elements in $[1,n/2-i]\cup [n/2+i,n]$ provides no new information, as the probability that the comparison will land one way or the other is independent on which unlabeled element is $n/2+i$ and which is $n/2-i$. The only comparisons that provide any information are those that compare an unlabeled element to an element in $[n/2-i,n/2+i]$. So these comparisons are the only source of information, and each such comparison provides $O(1/\gamma^2)$ bits of information, by Lemma~\ref{lem:info}. 
\end{proof}

\begin{corollary}\label{cor:info}
Any algorithm that answers \textsc{Which-is-Which} correctly with probability $1/2+\delta$ makes $\Omega(\delta^2/\gamma^2)$ comparisons between elements in $\{n/2+i,n/2-i\}$ and $[n/2-i,n/2+i]$ in expectation.
\end{corollary}

Now, we want to show that any algorithm that correctly accepts $n/2-i$ and rejects $n/2+i$ implies an algorithm that solves \textsc{Which-is-Which} with good probability while making the same comparisons. This would then imply that the original algorithm made many comparisons between elements in $\{n/2+i,n/2-i\}$ as well. 

\begin{lemma}
Let $\mathcal{A}$ denote an algorithm for \select\ that makes $X_i$ mistakes in expectation on the elements $\{n/2+i,n/2-i\}$. Then there is an algorithm $\mathcal{A}'$ for \textsc{Which-is-Which} that makes exactly the same comparisons as $\mathcal{A}$ and is correct with probability $1-X_i/2$.
\end{lemma}
\begin{proof}
We define $\mathcal{A}'$ as follows. Run $\mathcal{A}$ on the input provided. If both of the unlabeled elements are accepted, or both are rejected, output a random answer. If one is accepted and the other is rejected, guess that the accepted element is $n/2-i$. $\mathcal{A}'$ will then be correct exactly half of the time that $\mathcal{A}$ makes one mistake, all of the time that $\mathcal{A}$ makes zero mistakes, and none of the time that $\mathcal{A}$ makes two mistakes. So $\mathcal{A}'$ is wrong with probability exactly $X_i/2$.
\end{proof}

\begin{corollary}\label{cor:bound}
Any algorithm that makes $X_i<1$ mistakes in expectation on elements in $\{n/2+i,n/2-i\}$ makes $\Omega( ((1-X_i)/\gamma)^2)$ comparisons between elements in $\{n/2+i,n/2-i\}$ in expectation.
\end{corollary}

Finally, we are ready to apply a similar approach to the previous lower bounds. Note that, because the input elements are a priori indistinguishable, each of the $\binom{n}{2}$ possible comparisons are made with probability $\approx 2d/n$. As there are $4i$ comparisons that compare elements in $\{n/2+i,n/2-i\}$ to elements in $[n/2-i,n/2+i]$, the expected number of such comparisons is $\approx 8di/n$. Applying the contrapositive of Corollary~\ref{cor:bound}, we see that such an algorithm cannot possibly make fewer than $1 - \sqrt{Kdi\gamma^2/n}$ mistakes in expectation, for some absolute constant $K$. Now, we can take similar sums to the previous lower bounds to complete the proof. 

$$\text{Error} \geq \sum_{i = 1}^{n/Kd\gamma^2} (1-\sqrt{Kd\gamma^2 i/n})i^c$$
$$ \geq \sum_{i=1}^{n/Kd\gamma^2} i^c - \sum_{i=1}^{n/Kd\gamma^2} \sqrt{Kd\gamma^2/n} \cdot i^{(2c+1)/2}$$
$$ \geq (n/Kd\gamma^2)^{c+1}/(c+1) - \sqrt{Kd\gamma^2/n} \cdot 2(n/Kd\gamma^2)^{(2c+3)/2}/(2c+3)$$
$$ = \left(\frac{n}{Kd\gamma^2}\right)^{c+1}\left(\frac{1}{c+1}-\frac{2}{2c+3}\right)$$
$$ = \frac{n^{c+1}}{(Kd\gamma^2)^{c+1}(c+1)(2c+3)} = 
\Omega\left(\left(\frac{n}{d\gamma^2}\right)^{c+1}\right)$$.


\section{Proofs for Adaptive Algorithms in Noiseless Model}
\label{app:multinoiseless}
\subsection{Two Rounds}
\subsubsection{Upper Bound}

\begin{algorithm}[ht]
        \caption{Two-round algorithm for the noiseless model making $rn$ queries in round one}
    \begin{algorithmic}[1]\label{alg:iter2round}
	\STATE Variables $S_i$ will denote the $i^{th}$ refined skeleton set. Initialize $S_0 = [n]$.
	\STATE Variables $T_i$ will denote the pivots for $S_i$ (smaller sets will have more pivots). 
	\STATE Variables $a_i$, $b_i$ will denote two elements of $T_i$, and we'll denote $A_i = [a_{i-1},b_{i-1}]\cap S_i$. Initialize $a_0 = 1,b_0 = n$.
	\STATE Variables $k_i$ will be such that the $k_i^{th}$ element of $A_i$ is very close to the $k_{i-1}^{th}$ element of $A_{i-1}$. Initialize $k_0 = k_1 = n/2$. 
	\FOR{$i=1$ to $r$}
	\STATE Let $S_i$ be $n^{1-i/(2r+1)}$ random samples from $S_{i-1}$ (without repetition).
	\STATE Let $T_i$ be $n^{i/(2r+1)}$ random samples from $S_i$ (without repetition).
	\ENDFOR
	\STATE In round one, for all $i$, compare every element in $S_i$ to every element in $T_i$.
	\FOR{$i=1$ to $r$}
	\STATE Let $a_i$ be the largest element in $T_i$ that beats at most $k_i$ elements of $A_i$. Let $b_i$ be the smallest element of $T_i$ that beats at least $k_i$ elements of $A_i$. 
	\STATE Update $A_{i+1} = [a_i,b_i] \cap S_{i+1}$, and update $k_{i+1}$ so that the $(k_{i+1}|[a_i,b_i]\cap A_i|/|A_{i+1}|)^{th}$ element of $[a_i,b_i]\cap A_i$ is exactly the $k_i^{th}$ element of $A_i$. 
	\ENDFOR
	\STATE Let $x$ be an arbitrary element of $[a_r,b_r]$. 
	\STATE In round two, compare every element $j$ to $x$. If $j$ beats $x$, accept. Otherwise, reject.
             \end{algorithmic}
\end{algorithm}

\begin{prevproof}{Theorem}{thm:iter2round} 
First of all, it's easy to see that Algorithm~\ref{alg:iter2round} has 2 rounds and makes $r\cdot n$ comparisons in the first round and $n$ comparisons in the second round.

Let's show step 12 and step 13 in Algorithm~\ref{alg:iter2round} are valid. Notice that in step 9, the algorithm compares every element in $S_i$ with every element in $T_i$. Since $S'_i$ is a subset of $S_i$ , the algorithm knows the comparison results between every element in $S'_i$ and every element in $T_i$. Therefore the algorithm knows the ranking of each element of $T_i$ in $S'_i$. Thus the algorithm can find $a_i$ and $b_i$ according to the rankings of elements of $T_i$ in $S'_i$. By the comparison results between elements in $T_i$ and elements in $S'_i$, the algorithm can also figure out $A_i$. Finally, by looking at the rankings of $a_i$ and $b_i$ in $S'_i$, the algorithm knows the ranking of the $(p_{i-1}\cdot |S'_i|)$-th element of $S'_i$ in $A_i$.

We then prove the following several lemmas to show that $x$ is a good approximation of the median. 
\begin{lemma}
\label{lem:iter2r1}
For each $i$ in $1,...,r$, with probability at least $1- 2e^{-n^{\varepsilon}/2}$, $|A_i| \leq n^{1 - 2i/(2r+1) + \varepsilon}$. 
\end{lemma}
\begin{proof}
Let $A'$ be the set of elements ranked from $\max\{(p_{i-1}\cdot |S'_i|) - n^{1-2i/(2r+1) +\varepsilon}/2, 1\}$ to $(p_{i-1}\cdot |S'_i|)$ in $S'_i$. $A'$ has size at most $n^{1-2i/(2r+1) +\varepsilon}/2$. For each sample in $T_i$, the probability that it is in $A'$ is $\frac{n^{1-2i/(2r+1) +\varepsilon}}{2n^{1-i/(2r+1)}}$. So the probability that $T_i \cap A' = \emptyset$ is at most 
\[
(1-\frac{n^{1-2i/(2r+1) +\varepsilon}}{2n^{1-i/(2r+1)}})^{n^{i/(2r+1)}} \leq e^{-\frac{n^{1-2i/(2r+1) +\varepsilon}}{2n^{1-i/(2r+1)}}\cdot n^{i/(2r+1)}} = e^{-n^{\varepsilon}/2}.
\]
Similarly, let $A''$ be the set of elements ranked from $(p_{i-1}\cdot |S'_i|)$ to $\min\{(p_{i-1}\cdot |S'_i|) + n^{1-2i/(2r+1) +\varepsilon}/2, |S'_i|\}$ to in $S'_i$. With probability at most $e^{-n^{\varepsilon}/2}$, $T_i \cap A'' = \emptyset$. When $T_i \cap A' \neq \emptyset$ and $T_i \cap A'' \neq \emptyset$, $A_i$ can only contain elements that ranked between $\max\{(p_{i-1}\cdot |S'_i|) - n^{1-2i/(2r+1) +\varepsilon}/2, 1\}$ and $\min\{(p_{i-1}\cdot |S'_i|) + n^{1-2i/(2r+1) +\varepsilon}/2, |S'_i|\}$ in $S'_i$. So it is clear that $|A_i| \leq n^{1 - 2i/(2r+1) + \varepsilon}$. By union bound, with probability at least $1- 2e^{-n^{\varepsilon}/2}$, $|A_i| \leq n^{1 - 2i/(2r+1) + \varepsilon}$.
\end{proof}

\begin{lemma}
\label{lem:iter2r2}
For each $i$ in $2,...,r$, with probability at least $ 1 - e^{-n^{1/(2r+1) + \varepsilon}/(8r^2)}$, either $|A_{i-1}| \leq n^{2/(2r+1) + \varepsilon }$ or $|S'_i| \geq (1-1/(2r))|A_{i-1}|/(n^{1/(2r+1)})$. 
\end{lemma}
\begin{proof}
Let's consider the case when $|A_{i-1}| \geq n^{2/(2r+1) + \varepsilon}$. For each element in $S_{i-1}$, the probability that it is in $A_{i-1}$ is $|A_{i-1}|/|S_{i-1}|$. So $\E[|S'_i|] = (|A_{i-1}|/|S_{i-1}|) \cdot |S_i| = |A_{i-1}|/n^{1/(2r+1)}$. By Chernoff bound, the probability that $|S'_i| < (1-1/(2r))|A_{i-1}|/(n^{1/(2r+1)})$ is at most $e^{-\frac{|A_{i-1}|}{8r^2n^{1/(2r+1)}}} \leq e^{-n^{1/(2r+1) + \varepsilon}/(8r^2)}$. 
\end{proof}

\begin{lemma}
\label{lem:iter2r3}
Let $q = 1/(1-1/(2r))$. For each $i$ in $1,...,r$, with probability at least $1- 2e^{-n^{\varepsilon}/2} - e^{-n^{1/(2r+1) + \varepsilon}/(8r^2)}- 2e^{-n^{\varepsilon}/4}$, the $\min\{p_i\cdot |A_i| + 2r(q^{r+1-i} - 1)n^{(r+1-i)/(2r+1)+\varepsilon}, |A_i|\}$-th element of $A_i$ is at most the $\min\{p_{i-1} \cdot |A_{i-1}| +2r(q^{r-i + 2} - 1)n^{(r+2-i)/(2r+1)+\varepsilon},|A_{i-1}|\}$-th element of $A_{i-1}$ and the $\max\{p_i\cdot |A_i| - 2r(q^{r+1-i} - 1)n^{(r+1-i)/(2r+1)+\varepsilon}, 1\}$-th element of $A_i$ is at least the $\max\{p_{i-1} \cdot |A_{i-1}| - 2r(q^{r-i + 2} - 1)n^{(r+2-i)/(2r+1)+\varepsilon}, 1\}$-th element of $A_{i-1}$.
\end{lemma}
\begin{proof}
First from how we choose $A_i$ we know that  the $\min\{p_i\cdot |A_i| + 2r(q^{r-i+1} - 1)n^{(r+1-i)/(2r+1)+\varepsilon}, |A_i|\}$-th element of $A_i$ is at most the $\min\{p_{i-1}\cdot |S'_i| + 2r(q^{r-i+1} - 1)n^{(r+1-i)/(2r+1)+\varepsilon}, |S'_i|\}$-th element of $S'_i$. When $p_{i-1}\cdot|A_{i-1}| +2r(q^{r-i + 2} - 1)n^{(r+2-i)/(2r+1)+\varepsilon} \geq |A_{i-1}|$, the lemma is trivial. So let's consider the case that $p_{i-1}\cdot|A_{i-1}| +2r(q^{r-i + 2} - 1)n^{(r+2-i)/(2r+1)+\varepsilon} < |A_{i-1}|$.

By Lemma \ref{lem:iter2r1}, Lemma \ref{lem:iter2r2} and Union bound, we know that with probability at least $1- 2e^{-n^{\varepsilon}/2} - e^{-n^{1/(2r+1) + \varepsilon}/(8r^2)}$, we have $|A_{i-1}| \leq n^{1 - (2i-2)/(2r+1) + \varepsilon}$ and also either $|A_{i-1}| \leq n^{2/(2r+1) + \varepsilon }$ or $|S'_i| \geq |A_{i-1}|/(qn^{1/(2r+1)})$. If $|A_{i-1}| \leq n^{2/(2r+1) + \varepsilon}$ , we have $2r(q^{r-i + 2} - 1)n^{(r+2-i)/(2r+1)+\varepsilon} \geq |A_{i-1}|$ and the lemma becomes trivial. So let's consider the case when $|A_{i-1}| \leq n^{1 - (2i-2)/(2r+1) + \varepsilon}$  and $|S'_i| \geq |A_{i-1}|/(qn^{1/(2r+1)})$. 

Let $X$ denote the number of elements in $S'_i$ that is at most the $p_{i-1} \cdot |A_{i-1}| +2r(q^{r-i + 2} - 1)n^{(r+2-i)/(2r+1) +\varepsilon}$-th element of $A_{i-1}$. Then we have
\begin{eqnarray*}
\E[X] &=& (p_{i-1} \cdot |A_{i-1}| +2r(q^{r-i + 2} - 1)n^{(r+2-i)/(2r+1) +\varepsilon} ) \cdot \frac{|S'_i|}{|A_{i-1}|} \\
 &\geq& p_{i-1}\cdot |S'_i| + 2r(q^{r-i+1} - 1)n^{(r+1-i)/(2r+1)+\varepsilon } + n^{(r+1-i)/(2r+1)+\varepsilon }.
\end{eqnarray*}
By Chernoff bound,
\begin{eqnarray*}
&& Pr[ X \leq  p_{i-1} \cdot |S'_i| + 2r(q^{r-i+1} - 1)n^{(r+1-i)/(2r+1)+\varepsilon }] \\
&\leq& \exp(-(n^{(r+1-i)/(2r+1)+\varepsilon }/|A_{i-1}|)^2\cdot |S'_i| / 2 )\\
&\leq& \exp(-n^{2(r+1-i)/(2r+1)+2\varepsilon } /(2 \cdot (qn^{1/(2r+1)}) \cdot n^{1 - (2i-2)/(2r+1) + \varepsilon})) \\
&=& e^{-n^{\varepsilon}/4} 
\end{eqnarray*}
So when $|A_{i-1}| \leq n^{1 - (2i-2)/(2r+1) + \varepsilon}$  and $|S'_i| \geq |A_{i-1}|/(qn^{1/(2r+1)})$, with probability at least $1- e^{-n^{\varepsilon}/4}$, the $\min\{p_i\cdot |A_i| + 2r(q^{r+1-i} - 1)n^{(r+1-i)/(2r+1)+\varepsilon}, |A_i|\}$-th element of $A_i$ is at most the $\min\{p_{i-1} \cdot |A_{i-1}| +2r(q^{r-i + 2} - 1)n^{(r+2-i)/(2r+1)+\varepsilon},|A_{i-1}|\}$-th element of $A_{i-1}$. We can use the same argument to show that when $|A_{i-1}| \leq n^{1 - (2i-2)/(2r+1) + \varepsilon}$  and $|S'_i| \geq |A_{i-1}|/(qn^{1/(2r+1)})$, with probability at least $1- e^{-n^{\varepsilon}/16}$, the $\max\{p_i\cdot |A_i| - 2r(q^{r+1-i} - 1)n^{(r+1-i)/(2r+1)+\varepsilon}, 1\}$-th element of $A_i$ is at least the $\max\{p_{i-1} \cdot |A_{i-1}| - 2r(q^{r-i + 2} - 1)n^{(r+2-i)/(2r+1)+\varepsilon}, 1\}$-th element of $A_{i-1}$. 
\end{proof}
By applying Lemma \ref{lem:iter2r3} for $r$ times, we know that with probability at least $1- r(2e^{-n^{\varepsilon}/2} - e^{-n^{1/(2r+1) + \varepsilon}/(8r^2)}- 2e^{-n^{\varepsilon}/4})$, the $\min\{p_i\cdot |A_r| + n^{1/(2r+1)+\varepsilon}, |A_r|\}$-th element of $A_r$ is at most the $n/2 +2r(q^{r + 1} - 1)n^{(r+1)/(2r+1)+\varepsilon}$-th element of $A_0$ and the $\max\{p_i\cdot |A_r| - n^{1/(2r+1)+\varepsilon}, 1\}$-th element of $A_r$ is at least the $n/2 - 2r(q^{r + 1} - 1)n^{(r+1)/(2r+1)+\varepsilon}$-th element of $A_0$ and $|A_r| \leq n^{1/(2r+1)+\varepsilon}$. So we know $x$ is between $n/2 - 2r(q^{r + 1} - 1)n^{(r+1)/(2r+1)+\varepsilon}$ and $n/2 +2r(q^{r + 1} - 1)n^{(r+1)/(2r+1)+\varepsilon}$. Therefore the $c$-weighted error is at most $(2r(q^{r + 1} - 1)n^{(r+1)/(2r+1)+\varepsilon} )^{c+1} \leq (8rn^{(r+1)/(2r+1)+\varepsilon})^{c+1} $.
\end{prevproof}

By taking $r=\log n$ and $n^{\varepsilon} = 8\log^3 n $ we get the following corollary of Theorem \ref{thm:iter2round}:

\begin{corollary}
\label{cor:iter2r}
There exists a 2-round algorithm which takes $n\log n$ queries in each round and has number of mistakes at most $128n^{1/2} \log^4 n$ with probability $1-O(\log n / n^2)$. 
\end{corollary}

\subsubsection{Lower Bound}

\begin{prevproof}{Theorem}{thm:lowerbound2round} 
We are going to show the $c$-weighted lower bound of the algorithm on the uniform distribution of orders. Therefore, it is sufficient to prove this theorem for deterministic algorithms since we are considering the expected number of mistakes the algorithm made on the uniform distribution of orders. So without the loss of generality, let's assume the 2-round algorithm is a deterministic algorithm.

Suppose the algorithm labels all the elements as $1,...,n$. Let their rankings be $a_1, ... ,a_n$. Then $a_1,...,a_n$ is a uniformly random permutation of $1,...,n$. For any instance of run of the 2-round algorithm, we define several terms. Let $S = \{i| n/2-n/(32d) \leq a_i \leq n/2+n/(32d)\}$. Let $s = |S| = n/(16d)$. Let $e(S)$ be the number of comparisons in the first round whose elements are both in $S$. Then we have 
\[
\E[e(S)] = dn \cdot \frac{n/(16d)}{n} \cdot \frac{n/(16d)-1}{n-1} < n / (256d). 
\]
By averaging argument, with probability at least $3/4$, $e(S) \leq n/(64d)$. Let $A$ be the set of elements that are in $S$ and have not been compared with any other elements in $S$ in the first round of the algorithm. Then we have $|A| \geq s - 2\cdot e(S)$. If $e(S) \leq n/(64d)$, we have $|A| \geq n/(32d) = s /2$. 

Let's fix comparison results $S$. Then $A$ is also fixed by $S$. Let's focus on the case when $|A| \geq s / 2$. Let $A'$ be an arbitrary subset of $A$ of size $s/2$. Let $i_1, ...,i_{s/2}$ be elements in $A'$. Fix the rankings of elements in $\{1,...,n\} - A'$. It's easy to see that $a_{i_1},...,a_{i_{s/2}}$ will be $s/2$ distinct random elements from $[n/2 - n/(32d),n/2+n/(32d)]$. Also, since we have fixed $S$ and $A'$ and rankings of elements in $\{1,...,n\}-A'$, the result seen by the algorithm after the first round is fixed. So comparisons made by the algorithm in the second round is determined. Now let $T = \{i | n/2 - t\sqrt{n}/2 \leq a_i < n/2 + t\sqrt{n}/2\}$. Here $t$ will be determined later. Let $B = A' \cap T$. We have $\E[|B|] = \frac{|T|}{s} \cdot |A'| = t\sqrt{n}/2$. We are going to prove the following claim which we will use to show that in expectation there's constant fraction of elements in $B$ that are not compared with elements in $T$ even after the second round. 
\begin{claim}
Fix $S$, $A'$ and rankings of elements in $\{1,...,n\}-A'$. For each comparison $(x,y)$ in the second round, the probability that $x \in T$ and $y \in B$ is at most $ \frac{256t^2d^{3/2}}{\sqrt{n}}$. 
\end{claim}
\begin{proof}
We prove this claim by considering two cases.
\begin{enumerate}
\item $x \in A'$: In this case, we have both $x$ and $y$ are in $B$. The probability that this is true is 
\[
\frac{|T|}{s} \cdot \frac{|T|-1|}{s-1} \leq (\frac{|T|}{s})^2 = 256t^2d^2/ n \leq  \frac{256t^2d^{3/2}}{\sqrt{n}}. 
\]
\item $x \in S - A'$: In this case, we need the following lemma:
\begin{lemma}
If we pick $s/2$ random distinct elements from a set $S$ of $s$ elements. Then the $k$-th element of these $s/2$ elements is ranked between $s/2 - t\sqrt{n}/2$ and $s/2 + t\sqrt{n}/2$ among $S$ with probability at most $\frac{4t\sqrt{n}}{\sqrt{s}}$. 
\end{lemma}
\begin{proof}
The probability that $k$-th element of these $s/2$ elements is ranked between $s/2 - t\sqrt{n}/2$ and $s/2 + t\sqrt{n}/2$ among $S$ is
\[
\sum_{u = s/2 - t\sqrt{n}/2}^{s/2 + t\sqrt{n}/2} \frac{\binom{u-1}{k-1}\binom{s-u}{s/2 - k }}{\binom{s}{s/2}}. 
\]
It is easy to see that $\binom{u-1}{k-1}\binom{s-u}{s/2 - k } \leq \binom{s/2}{s/4}^2$. Therefore, we have 
\[
\sum_{u = s/2 - t\sqrt{n}/2}^{s/2 + t\sqrt{n}/2} \frac{\binom{u-1}{k-1}\binom{s-u}{s/2 - k }}{\binom{s}{s/2}} \leq t\sqrt{n} \cdot \frac{\binom{s/2}{s/4}^2}{\binom{s}{s/2}}. 
\]
By Stirling approximation, we know that $ \sqrt{2\pi} \leq \frac{n!}{n^{n+1/2} e^{-n}} \leq e$. Then we have 
\[
 t\sqrt{n} \cdot \frac{\binom{s/2}{s/4}^2}{\binom{s}{s/2}} \leq  t\sqrt{n} \cdot \frac{e^4}{\sqrt{2\pi}^5} \cdot \frac{\left((s/2)^{s/2} \sqrt{s/2}\right)^4}{\left((s/4)^{s/4}\sqrt{s/4}\right)^4 s^s \sqrt{s}} \leq \frac{4t\sqrt{n}}{\sqrt{s}}. 
\]
\end{proof}
Then we have
\[
Pr[x \in T, y \in B] = Pr[ x \in T] \cdot Pr[y \in B|x \in T] \leq \frac{4t\sqrt{n}}{\sqrt{s}} \cdot \frac{t\sqrt{n}}{s}  = \frac{256t^2d^{3/2}}{\sqrt{n}} .
\]
\end{enumerate}
\end{proof}
Since in the second round, we have $dn$ comparisons, the expected number of comparisons that is in the form $(x \in T, y \in B)$ is at most $dn \cdot  \frac{256t^2d^{3/2}}{\sqrt{n}}$. By picking $t = \frac{1}{2048 d^{5/2}}$, we have $dn \cdot  \frac{256t^2d^{3/2}}{\sqrt{n}} \leq t\sqrt{n} /8$ . Let $C \subseteq B$ be the set of elements that have not been compared with other elements in $T$. Then we have $\E[|C|] \geq \E[|B|] - 2 \cdot  (t\sqrt{n} /8) \geq t\sqrt{n}/4$. 

Finally, let's fix the rankings of elements in $A' - C$. Suppose the elements in $C$ are $j_1, ..., j_{|C|}$, then $a_{j_1},...,a_{j_{|C|}}$ are $|C|$ distinct random elements between $n/2 - t\sqrt{n}/2$ and $n/2 + t\sqrt{n}/2$. Also when we fix the rankings of $\{1,...,n\}-C$, the comparison results got by the algorithms in both the first round and the second round are determined. So the output of the algorithm is determined and therefore the algorithm has expected number of mistakes $|C|/2$. 

To sum up,  any 2-round algorithm (maybe randomized) which makes $dn$ comparisons in each round has expected number of mistakes at least
\[
\frac{3}{4} \cdot (t\sqrt{n}/4) \cdot (1/2) = \Omega(n^{1/2}/d^{5/2}).
\]
\end{prevproof}

To complete the proof for $c > 0$, we observe that the best case scenario, conditioned on making $\Omega(\sqrt{n}/d^{5/2})$ mistakes in expectation is that the algorithm is correct on every element outside of $[n/2-\Omega(\sqrt{n}/d^{5/2}),n/2+\Omega(\sqrt{n}/d^{5/2})]$, and incorrect on every element in this window. This is because elements outside this window contribute more to the $c$-weighted error than elements inside. We have $\sum_{i = 1}^{\Omega(\sqrt{n}/d^{5/2})} i^c = \Omega((\sqrt{n}/d^{5/2})^{c+1})$ as desired.

\subsection{Three Rounds}

\begin{algorithm}[ht]
        \caption{Three-round algorithm for the noiseless model making $O(n\log^8 n)$ queries}
    \begin{algorithmic}[1]\label{alg:3round}
		\STATE Define $t = 256\sqrt{n}\log^4 n$. 
		\STATE Simultaneously do the following two steps in two rounds:
		\STATE Run Algorithm~\ref{alg:iter2round} on the given input with $2t$ additional dummy $0$ elements and $r = \log n$ (i.e. treat these dummy elements as smaller than any of the real elements, and compare them to each other arbitrarily). Let $A^{(1)}$ denote the set that was accepted, and $R^{(1)}$ the set that was rejected.
		\STATE Run Algorithm~\ref{alg:iter2round} on the given input with $2t$ additional dummy $n+1$ elements and $r = \log n$ (i.e. treat these dummy elements as larger than any of the real elements, and compare them to each other arbitrarily). Let $A^{(2)}$ denote the set that was accepted, and $R^{(2)}$ the set that was rejected.
		\STATE Accept all elements in $A^{(2)}$. Reject all elements in $R^{(1)}$. Let $U = R^{(2)} \cap A^{(1)}$ be the remaining elements.
		\IF{$|U| > 4t$}
		\STATE Make random decisions for all elements in $U$. 
		\ELSE
		\STATE Compare every element in $U$ to every other element in $U$. Let $x$ denote the $(n/2-|R^{(2)}|)^{th}$ element of $U$. 
		\STATE Accept every element in $U$ that beat $x$, and reject every element in $U$ that $x$ beat.
		\ENDIF
		
             \end{algorithmic}
\end{algorithm}

We first provide a simple observation that is useful in the analysis of both Algorithm \ref{alg:3round} and Algorithm \ref{alg:4rounds}. 
\begin{observation}
\label{ob:dum}
Suppose we have an algorithm $A$ on \partition\ of $n+2t$ elements and $A$ makes mistakes only on elements between $n/2$ and $n/2 + 2t$. If we want to solve \partition\ on $n$ element, we can add $2t$ dummy elements which we assume they are smaller than other elements and run $A$ on these $n+2t$ elements. By doing this, if $A$ claims some element $x$ is smaller than the median of the $n+2t$ elements, then $x$ is indeed smaller than the median of the original $n$ elements. 
\end{observation}

\begin{prevproof}{Theorem}{thm:3round} 
Let's first check the number of comparisons made by Algorithm \ref{alg:3round}. In the first two rounds, by Corollary \ref{cor:iter2r}, Algorithm \ref{alg:3round} makes at most $4(n+2t)\log (n+2t) = O(n\log n)$ comparisons. In the third round, since $U$ has size at most $4t$, Algorithm \ref{alg:3round} makes at most $(4t)^2 = O(n \log^8n)$ comparisons. 

Now let's analyze the correctness of Algorithm \ref{alg:3round}. Implicitly from the proof of Theorem \ref{thm:iter2round}, we know that with probability $1- O(\log n / n^2)$, Algorithm \ref{alg:iter2round} used in step 3 and step 4 makes mistakes only for elements that are most $t$ away from the median of those $n+2t$ elements. In this case, by Observation \ref{ob:dum}, we know that all elements in $R^{(1)}$ are smaller than the median of the original $n$ elements and all elements in $A^{(2)}$ are larger than the median of the original $n$ elements. Also we have $|R^{(1)}|, |A^{(2)}| \geq (n+2t)/2-t-2t = n/2 - 2t$. So $|U| \leq 4t$ and  Algorithm \ref{alg:3round} will also output correctly on elements in $U$. To sum up, with probability $1- O(\log n / n^2)$, Algorithm \ref{alg:3round} outputs a partition with \emph{zero} $c$-weighted error.
\end{prevproof}

\subsection{Four Rounds}

\begin{algorithm}[ht]
        \caption{Four-round algorithm for the noiseless model making $O(n)$ queries}
    \begin{algorithmic}[1]\label{alg:4rounds}
		\STATE Pick any $\epsilon \in [0,1/18]$. Define $t = n^{2/3+\epsilon}$. 
		\STATE Simultaneously do the following two steps in two rounds:
		\STATE Run Algorithm~\ref{alg:iter2round} on the given input with $2t$ additional dummy $0$ elements and $r = 1$ (i.e. treat these dummy elements as smaller than any of the real elements, and compare them to each other arbitrarily). Let $A^{(1)}$ denote the set that was accepted, and $R^{(1)}$ the set that was rejected.
		\STATE Run Algorithm~\ref{alg:iter2round} on the given input with $2t$ additional dummy $n+1$ elements and $r = 1$ (i.e. treat these dummy elements as larger than any of the real elements, and compare them to each other arbitrarily). Let $A^{(2)}$ denote the set that was accepted, and $R^{(2)}$ the set that was rejected.
		\STATE Accept all elements in $A^{(2)}$. Reject all elements in $R^{(1)}$. Let $U = R^{(2)} \cap A^{(1)}$ be the remaining elements.
		\IF{$|U| > 4t$}
		\STATE Make random decisions for all elements in $U$. 
		\ELSE
		\STATE Pick a random set $V \subset U$ with $|V| = n/|U|$. 
		\STATE In round three, compare every element of $V$ with every element of $U$. 
		\STATE Let $x$ be the largest element in $V$ that beats at most $n/2-|R^{(2)}|$ elements of $U$. Let $y$ be the smallest element in $V$ that beats at least $n/2-|R^{(2)}|$ elements of $U$. 
		\STATE Reject every element in $U$ that is beaten by $x$. Accept every element in $U$ that beats $y$. Denote by $W$ the remaining elements.
		\IF{$|W| >n^{1/3+3\epsilon}$}
		\STATE Make random decisions for every element in $W$.
		\ELSE
		\STATE In round four, compare every element in $W$ to every other element in $W$.
		\STATE Let $R^*$ denote the set of elements that have already been rejected. Let $z$ denote the $(n/2-|R^*|)^{th}$ element of $W$. 
		\STATE Accept all elements of $W$ that beat $z$, and reject all elements that are beaten by $z$.
		\ENDIF
		\ENDIF
		
             \end{algorithmic}
\end{algorithm}

\begin{prevproof}{Theorem}{thm:4rounds} 
First it's easy to check that the number of comparisons made Algorithm \ref{alg:4rounds} is $O(n)$. 

Now let's analyze the correctness of Algorithm \ref{alg:4rounds}. Notice that the first two rounds of  Algorithm \ref{alg:4rounds} are almost the same as the first two rounds of Algorithm \ref{alg:3round}, except the fact that they use different version of Algorithm \ref{alg:iter2round}. So we can use the same argument as the argument in the proof of Theorem \ref{thm:3round}. Then we get with probability $1-4e^{-n^{\varepsilon}/2} - 2e^{-n^{1/3+\varepsilon}/8} - 4e^{-n^{\varepsilon} / 4} = 1 - e^{-\Omega(n^{\varepsilon})}$, all elements in $R^{(1)}$ are smaller than the median of the original $n$ elements, all elements in $A^{(2)}$ are larger than the median of the original $n$ elements and $|U| \leq 4t$. We call this as ``succeed in the first two rounds". 

Now let's consider the third round and analyze the probability that $|W| \leq n^{1/3 +3 \varepsilon}$ assuming we succeed in the first two rounds. Let $u$ be the size of $U$. We know $u \leq 4t$. If the followings two conditions were satisfied, $|W| \leq n^{1/3 +3 \varepsilon}$.
\begin{enumerate}
\item There is an element in $V$ which ranks between $u/2 - n^{1/3+3\varepsilon}/2$ and $u/2$. 
\item There is an element in $V$ which ranks between $u/2$ and $u/2 + n^{1/3+3\varepsilon}/2$. 
\end{enumerate}
The probability that each condition is not satisfied is at most 
\[
(1-\frac{n^{1/3+3\varepsilon}/2}{u})^v \leq e^{- \frac{n^{1/3+3\varepsilon}/2}{u} \cdot v} \leq e^{-\Omega(n^{\varepsilon})}.
\]
So with probability $1 - e^{-\Omega(n^{\varepsilon})}$, $|W| \leq n^{1/3 +3 \varepsilon}$. It's easy to see that if we succeed in the first two rounds and $|W| \leq n^{1/3 +3 \varepsilon}$, Algorithm \ref{alg:4rounds} outputs a partition with \emph{zero} $c$-weighted error. By Union bound, with probability $1 - e^{-\Omega(n^{\varepsilon})}$, Algorithm \ref{alg:4rounds} outputs a partition with \emph{zero} $c$-weighted error. 
\end{prevproof}

\section{Proofs for Adaptive Algorithms in Erasure and Noisy Models}
\label{app:multinoisy}

\subsection{Upper Bounds}
\subsubsection{Proposition~\ref{pro:blowup}}
\begin{prevproof}{Proposition}{pro:blowup}
Wlog, let's consider the problem we want to solve is \partition. We can do this because the constructions shown below do not depend on which problem we want to solve. Let the algorithm which solves \partition\ in the noiseless model be $A$. We prove the three statements in the proposition one by one. 
\begin{itemize} 
\item We construct algorithm $B_1$ in the erasure model that simulates $A$ in the following way: For each comparison $A$ makes, $B_1$ continues to make the same comparison until the result is not erased. Clearly $B_1$ solves \partition\ whenever $A$ solves \partition. And for each comparison $A$ makes, in expectation, $B_1$ makes $1/\gamma$ comparisons. So in total $B_1$ makes $Q/\gamma$ comparisons in expectation. And as $B_1$ is very adaptive, $B_1$'s expected round complexity is only bounded by $Q/\gamma$. 
\item We construct another algorithm $B_2$ in the erasure model that simulates $A$ and preserves the number of rounds: For each round of $A$, $B_2$ also starts a round of making each comparison $A$ made in that round $c_2 \cdot (\log Q + \log n)/\gamma$ times simultaneously. Here $c_2$ is some constant. By setting $c_2$ large enough, it's easy to see that with probability $1-1/\poly(n)$, all the comparisons $A$ wants to make are not erased. In this case, $B_2$ solves \partition\ whenever $A$ solves \partition. By Union bound, $B_2$ solves \partition\ with probability $p-1/\poly(n)$. 
\item We construct another algorithm $B_3$ in the noisy model that simulates $A$ and preserves the number of rounds: For each round of $A$, $B_3$ also starts a round of making each comparison $A$ made in that round $c_3 \cdot (\log Q + \log n)/\gamma^2$ times simultaneously and returns to $A$ the majority. Here $c_3$ is some constant. By setting $c_3$ large enough, by Chernoff bound and union bound, it's easy to see that with probability $1-1/\poly(n)$, $A$ gets all the comparison results correct. In this case, $B_3$ solves \partition\ whenever $A$ solves \partition. By Union bound, $B_3$ solves \partition\ with probability $p-1/\poly(n)$. 
\end{itemize}
\end{prevproof}

\subsubsection{Multi-round Algorithm in Erasure Model}
\begin{algorithm}[ht]
        \caption{Algorithm comparing every element in $[n]$ to pivot $x$ (subroutine of Algorithm \ref{alg:partitionerasure})}
    \begin{algorithmic}[1]\label{alg:rankerasure}
		\STATE Initialize $T = [n]$. $T$ will denote the subset of $[n]$ that hasn't been successfully compared to $x$ yet.
		\FOR{$i = 1$ to $\log^*(n)$}
		\STATE Compare each element of $T$ to $x$ $\frac{2n}{2^i|T|\gamma}$ times. 
		\STATE For each $j$, if any comparisons between $j$ and $x$ were not erased, remove $j$ from $T$.
		\ENDFOR
             \end{algorithmic}
\end{algorithm}

\begin{algorithm}[ht]
        \caption{Algorithm for \partition\ in erasure model}
    \begin{algorithmic}[1]\label{alg:partitionerasure}
		\STATE Set $t = n^{3/4}$.
		\STATE In rounds one thru four, simultaneously complete the following two steps.
		\STATE  Pick a random subset $S_1$ of size $n/\log n$ from $n$ elements and $2t$ additional dummy $0$ elements. Run the 4-round algorithm guaranteed by Corollary~\ref{cor:noisyalgs} for \select\ on $S_1$. Let $x_1$ denote the median output.
		\STATE Pick a random subset $S_2$ of size $n/\log n$ from $n$ elements and $2t$ additional dummy $0$ elements. Run the 4-round algorithm guaranteed by Corollary~\ref{cor:noisyalgs} for \select\ on $S_2$. let $x_2$ denote the median output.
		\STATE In rounds five thru $\log^*(n) + 4$, run Algorithm~\ref{alg:rankerasure} to compare every element in $[n]$ to both $x_1$ and $x_2$. 
		\STATE Reject all elements that are beaten by $x_1$, and accept all elements that beat $x_2$. Let $U$ denote the set of elements that are not answered.

		\STATE Add some dummy elements to $U$ to make $U$'s median be $n/2$. We can do this since we know the rank of $n/2$ in $U$ from the previous step. If $|U| > 12t$, fail. 
		\STATE In rounds $\log^*(n) + 5$ thru $\log^*(n)+8$, Run the 4-round algorithm guaranteed by Corollary~\ref{cor:noisyalgs} for \partition\ on $U$. Reject and accept elements in $U$ as the output of the 4-round algorithm.
             \end{algorithmic}
\end{algorithm}
\begin{prevproof}{Theorem}{thm:uberasure}
First, it is clear that Algorithm \ref{alg:partitionerasure} has $\log^*(n) + 8$ rounds. It is also easy to check Algorithm \ref{alg:partitionerasure} makes $O(n/\gamma)$ comparisons in total.

Similarly as the proof of Lemma \ref{lem:skeleton}, by using Lemma \ref{lem:rep}, with probability $1-1/\poly(n)$, $S_1$'s median will be at most $t$ elements away from the median of the $n+2t$ elements. Therefore $S_1$'s median will be ranked between $n/2-2t$ and $n/2$. Also by Corollary \ref{cor:noisyalgs}, we know that with probability $1-1/\poly(n)$, $x_1 =$ the median of $S_1$. So by Union bound, with probability $1-1/\poly(n)$, $x_1$ is between $n/2-2t$ and $n/2$. Similarly, with probability $1-1/\poly(n)$, $x_2$ is between $n/2$ and $n/2+2t$. 

We prove the following lemma to show that after running Algorithm \ref{alg:rankerasure}, most elements have non-erased comparison results with $x_1$ and $x_2$:
\begin{lemma}
\label{lem:tower}
With probability $1-1/\poly(n)$, $|T| \leq t$ at the end of Algorithm \ref{alg:rankerasure}.
\end{lemma}
\begin{proof}
Let's prove the following claim 
\begin{claim}
If $|T| \leq \max(\frac{n}{(2 \uparrow\uparrow (i-1))\cdot 2^{i-1}},t)$ before the $i^{th}$ iteration, then with probability $1 - e^{-\Omega(t)}$, $|T| \leq \max(\frac{n}{(2 \uparrow\uparrow i)\cdot 2^i},t)$ after the $i^{th}$ iteration. Here $\uparrow\uparrow$ is the notation for power tower, the inverse of $\log^*$.
\end{claim}
\begin{proof}
First if $|T| \leq t$, the claim is trivial.

Now we assume $t \leq |T| \leq \frac{n}{(2 \uparrow\uparrow (i-1))\cdot 2^{i-1}}$ before the $i$-th iteration. Now each element in $T$ will get $\frac{2 \uparrow\uparrow (i-1)}{\gamma}$ comparisons with $x$. Then for each element in $T$, the probability that it stays in $T$ after the $i$-th iteration is 
\[
(1- \gamma)^{\frac{2 \uparrow\uparrow (i-1)}{\gamma}} \leq e^{-2 \uparrow\uparrow (i-1)} \leq \frac{1}{(2 \uparrow\uparrow i)\cdot 2^{i+1}}.
\]
So the expectation of $|T|$ after the $i$-th iteration is at most $\frac{n}{(2 \uparrow\uparrow i)\cdot 2^{i+1}}$. By Chernoff bound, with probability $1 - e^{-\Omega(t)}$, $|T|  \leq \max(\frac{n}{(2 \uparrow\uparrow i)\cdot 2^i},t)$.
\end{proof}
By using this claim inductively and Union bound, we know that with probability $1 - \log^*(n) \cdot e^{-\Omega(t)} \geq 1 - 1/\poly(n)$, $|T| \leq \max (\frac{n}{(2\uparrow\uparrow \log^*(n))\cdot 2^{\log^*(n)}},t)  = t$. 
\end{proof}
By Lemma \ref{lem:tower}, we know that with probability $1-1/\poly(n)$, there are at most $2t$ elements that don't have non-erased comparison results with both $x_1$, $x_2$. Together with the fact that with probability $1-1/\poly(n)$, $x_1$ is between $n/2-2t$ and $n/2$, $x_2$ is between $n/2$ and $n/2+2t$, we know that with probability $1-1/\poly(n)$, $|U|$ will be at most $6t$ before adding dummy elements and at most $12t$ after adding dummy elements. In this case, it's easy to see that Algorithm \ref{alg:rankerasure} solves \partition\ correctly. 
\end{prevproof}

\subsubsection{Multi-Round algorithm for findMin in noisy model}
\begin{algorithm}[ht]
        \caption{Adaptive algorithm for \findMin\ in the noisy model}
    \begin{algorithmic}[1]\label{alg:findMin}
        \STATE Pick a random set $S$ of size $n/\log n$. Find the minimum of $S$ as $x$. (By Proposition 1, there's an algorithm with $\frac{n}{\gamma^2}$ comparisons to find minimum with probability $1-1/n$.) 
        \STATE Let $U$ be the set of all elements. Run the following loop, and fail whenever more than $\frac{2cn}{\gamma^2}$ comparisons are used. 
        \FOR{$k$ = 1 to $\log\log n$} 
        	\STATE Compare all elements in $U$ to $x$ for $\frac{c}{\gamma^2}$ times. 
	\STATE Remove any element $y \in U$ that beats $x$ for more than  $\frac{ck}{2\gamma^2}$ times in total.
	\ENDFOR
	\STATE If $|U| > n/\log n$, fail. Otherwise, find the minimum of $U$ and output it. (We use the same algorithm as step 1.)
             \end{algorithmic}
\end{algorithm}

\begin{prevproof}{Theorem}{thm:findMin}
Let's consider the following events. It will be clear that if all of them happen, then Algorithm \ref{alg:findMin} outputs the minimum. So we only have to show the probability that they all happen is at least $1-e^{-\Omega(c)}$.
\begin{enumerate}
\item The first event is that we find the minimum of $S$ correctly. This event happens with probability at least $1-1/n$. 
\item The second event is that the minimum of $S$ has rank $\leq \sqrt{n}$. This event happens with probability $1 - (1-1/\sqrt{n})^{n/\log n}\geq 1 - 1/n$
\item The third event is that for all $k$, after $k$-th iteration of step 2,  $|U| \leq  x\text{'s rank } + n \cdot e^{-ck}$. By Chernoff bound, the probability that an element that is larger than $x$ stays in $U$ after $k$-th iteration is at most $e^{-2ck}$.  So by Markov inequality, the probability that after $k$-th iteration $|U| >  x\text{'s rank } + n \cdot e^{-ck}$ is at most $e^{-ck}$. By union bound, the third event happens with probability at least $1-\sum_{k=1}^{\log \log n} e^{-ck} \geq 1- 2e^{-c}$. 
\item The fourth event is that the minimum is in $U$ after step 2. Similarly as the argument for the previous event, the probability that the minimum is not removed in the $k$-th iteration is $1-e^{-2ck}$. By union bound, the minimum is in $U$ after step 2 is at least $1 - \sum_{k=1}^{\log \log n} e^{-2ck} \geq  1- 2e^{-2c}$. 
\item The fifth event is that we find the minimum of $U$ correctly if $|U|\leq n/\log n$. This event happens with probability at least $1- 1/n$.  
\end{enumerate}
When the second event and the third event both happen, we will spend at most $\frac{c}{\gamma^2}(\sqrt{n} + n \cdot \sum_{k=1}^{\log \log n} e^{-ck})\leq \frac{2nc}{\gamma^2}$ comparisons in step 2 and after step 2, $|U|\leq\sqrt{n}+n \cdot e^{-c\log\log n} \leq n/\log n$. Then it's clear if all these five events happen, the algorithm outputs the minimum correctly. And by union bound, the probability that they all happen is at least $1 -3/n - 2e^{-c}-2e^{-2c} = 1-e^{-\Omega(c)}$. 
\end{prevproof}

\subsection{Lower Bounds}
\subsubsection{Lower Bound on findMin in Noisy Model}
\begin{prevproof}{Theorem}{thm:findMinlb}
We start with some notations used in the proof. Let's call the algorithm $A$. Let's consider $A$'s successful probability on the uniform distribution of orders. Without loss of generality, we can assume $A$ is deterministic since if $A$ is randomized, we can fix the randomness which makes $A$ achieves the highest successful probability. Assume $A$ labels the $n$ elements as $a_1,...,a_n$. Let $S$ be the comparison results. $S$ is just a $\frac{cn}{\gamma^2}$ bits string. We use $A(S)$ to denote the output of the algorithm which is one of $a_1,...,a_n$. Let $g(i)$ denote the number of comparisons involving element $a_i$ and an element that is smaller than $a_i$, and the comparison result is correct. Let $b(i)$ denote the number of comparisons involving element $a_i$ and an element that is smaller than $a_i$, and the comparison result is wrong. Since $A$ is deterministic, if we fix $S$, $A(S)$ is fixed. Let $W$ denote the event that there are at least $n/10+1$ $i$'s such that $g(i) - b(i) \leq \frac{9c}{\gamma}$. Let $\pi$ be a permutation of $a_1,...,a_n$. We use $\pi$ to denote an ordering of $a_1,...,a_n$ in the following way: if $a_i$ appears before $a_j$ in $\pi$, then it means $a_i < a_j$. Define $\pi(k)$ be the $k$-th element in permutation $\pi$.

Let's first show that $W$ happens with some constant probability. We need two new terms $p(i)$ and $q(i)$. They are set to be 0 before the algorithm's run and we will show how they change during the algorithm. For each comparison between $a_i$ and an element that is smaller than $a_i$, let the followings happen: 
\begin{enumerate}
\item With probability $\gamma$, $p(i)$ is increased by 1.
\item With probability $(1-\gamma)/2$, $q(i)$ is increased by 1. 
\item With probability $(1-\gamma)/2$, $q(i)$ is decreased by 1. 
\end{enumerate}
It's easy to check that $(p(i) + q(i))$'s have exactly the same distribution as $(g(i) - b(i))$'s.  From the rule above, we have 
\[
\E \sum_{i=1}^n p(i) = \frac{cn}{\gamma}
\] and 
\[
\E \sum_{i=1}^n q(i)^2  = \frac{(1-\gamma) cn}{\gamma^2}.
\]
The second equation comes from the fact that $((x-1)^2 + (x+1)^2)/2 =  x^2 + 1$. So after each comparison, in  expectation, $\sum_{i=1}^n q(i)^2$ is increased by 1. By Markov inequality, with probability at least $2/3$, $\sum_{i=1}^n p(i) \leq  \frac{3cn}{\gamma}$. With probability at least $2/3$, $\sum_{i=1}^n q(i)^2 \leq \frac{3cn}{\gamma^2}$. By union bound, with probability at least $1/3$, $\sum_{i=1}^n p(i) \leq  \frac{3cn}{\gamma}$ and $\sum_{i=1}^n q(i)^2 \leq \frac{3cn}{\gamma^2}$. In this case, at least $n/2$ $i$'s satisfy $p(i) \leq \frac{6c}{\gamma}$ and at least $2n/3$ $i$'s satisfy $q(i)^2 \leq \frac{9c}{\gamma^2}$. Then at least $n/6 \geq n/10 + 1$ $i$'s satisfy both $p(i) \leq \frac{6c}{\gamma}$ and $q(i)^2 \leq \frac{9c}{\gamma^2}$. Then for these $i$'s we have,
\[
p(i) + q(i) \leq p(i) + |q(i)| \leq \frac{6c}{\gamma} + \frac{3\sqrt{c}}{\gamma} \leq \frac{9c}{\gamma}. 
\]
So with probability at least $1/3$, there are at least $n/10 +1$ $i$'s such that $p(i)+q(i) \leq  \frac{9c}{\gamma}$. Since $(p(i) + q(i))$'s have exactly the same distribution as $(g(i) - b(i))$'s, we have $Pr[W] \geq 1/3$. 

After showing $Pr[W] \geq 1/3$, we are going to lower bound the probability that $A$ outputs incorrectly. For each permutation $\pi$, for each $i$ such that $a_i \neq \pi(1)$, define $\pi^i =( a_i, \pi(1),...,\pi(k-1),\pi(k+1),...,\pi(n))$ assuming $a_i = \pi(k)$. There are two properties about $\pi^i$:
\begin{enumerate}
\item Suppose we run $A$ on ordering $\pi$ and gets result $S$. Then $g(i)$ and $b(i)$ for all $i$ are fixed. Since $\pi^i$ only change the relative ordering between $a_i$ and the elements that are smaller than $a_i$, we have 
\[
\frac{Pr[S,\pi]}{Pr[S,\pi^i]} = \frac{Pr[S|\pi]}{Pr[S|\pi^i]} = \frac{(1/2 + \gamma)^{g(i)} (1/2-\gamma)^{b(i)}}{(1/2 + \gamma)^{b(i)} (1/2-\gamma)^{g(i)}} = (\frac{1/2+\gamma}{1/2-\gamma})^{g(i)-b(i)}.
\]
\item For each ordering $\pi'$ there are at most $n$ $\pi$'s satisfy the condition that there exists an $i$ such that $\pi^i = \pi'$. 
\end{enumerate}

By these two properties, we get the following:
\begin{eqnarray*}
 Pr[A \text{ outputs correctly}] &=& \sum_{\pi,S: \pi(1) = A(S)} Pr[\pi, S] \leq Pr[\neg W] +  \sum_{\pi,S :\pi(1) = A(S), W} Pr[\pi, S] \\
&\leq& Pr[\neg W] + \frac{10}{n} \sum_{\pi,S,i:\pi(1) = A(S), g(i)-b(i) \leq \frac{9c}{\gamma}, i\neq A(S), W}  Pr[\pi, S] \\
&\leq& Pr[\neg W] + \frac{10}{n} \sum_{\pi,S,i : \pi(1) = A(S), g(i)-b(i) \leq \frac{9c}{\gamma}, i \neq A(S), W} Pr[\pi^i, S] \cdot (\frac{1/2 + \gamma}{1/2 - \gamma})^{(g(i)-b(i))} \\
&\leq& Pr[\neg W] + \frac{10}{n} \cdot  \sum_{\pi,S: \pi(1)\neq S} n\cdot Pr[\pi, S] \cdot e^{\frac{9c}{\gamma} \cdot  (8\gamma)}\\
&=& Pr[\neg W] + 10 \cdot Pr[A \text{ outputs incorrectly}] \cdot e^{\frac{9c}{\gamma} \cdot  (8\gamma)}\\
&\leq& 2/3 + 10e^{72c} Pr[A \text{ outputs incorrectly}].
\end{eqnarray*}
Therefore $Pr[A \text{ outputs incorrectly}] \geq e^{-72c}/33$. 
\end{prevproof}

We provide the following lemma to show that \select\ is easier than \partition, even when approximation is involved (we proved Lemma~\ref{lem:relation} only in the exact case). And then lower bounds of \select\ also apply against \partition.
\begin{lemma}
Let $c>2$ be some constant. In any of the three models, suppose there's an algorithm $A$ on $n$ elements that makes $m$ comparisons and has at most $t$ mistakes on \partition\ with probability $p$. Then there's an algorithm $B$ on $n-ct$ elements that obtains a $ct/2$-approximation for \select\ with probability at least $p^2 \cdot \frac{c-2}{c+2}$. 
\end{lemma}

\begin{proof}
$B$ has the following steps:
\begin{enumerate}
\item Let the $n-ct$ elements be $a_1,...,a_{n-ct}$. 
\item Generate $ct$ dummy elements $b_1 < \cdots < b_{ct}$ which are smaller than any element in  $a_1,...,a_{n-ct}$. Run $A$ on $a_1,...,a_{n-ct}$ and $b_1,...,b_{ct}$. 
\item Generate $ct$ dummy elements $b'_1 < \cdots < b'_{ct}$ which are larger than any element in  $a_1,...,a_{n-ct}$. Run $A$ on $a_1,...,a_{n-ct}$ and $b'_1,...,b'_{ct}$. 
\item Let $S \subseteq \{a_1,...,a_{n-ct}\}$ be the of elements that are said to be above the median in the second step and below the median in the third step. Output an random element in $S$. 
\end{enumerate}
Let $S'$ be the elements that are at most $ct/2$ away from the median of $\{a_1,...,a_{n-ct}\}$. If $A$ does not make any mistakes, it's easy to check that $S = S'$. Now let's consider  the case that $A$ has at most $t$ mistakes in both the second and the third step. The probability of this to happen is $p^2$. In this case, we have $|S \cap S'| \geq (c-2)t$ and $|S| \leq (c+2)t$. It means in this case $B$ outputs an element that is $ct/2$ away from the median with probability at least $\frac{(c-2)t}{(c+2)t} = \frac{c-2}{c+2}$. In total, $B$ outputs an element that is $ct/2$ away from the median with probability at least $p^2 \cdot \frac{c-2}{c+2}$. 
\end{proof}
\subsubsection{Lower Bound for Select in Erasure Model}
\begin{prevproof}{Theorem}{thm:lberasure}
Let $A$ be the algorithm we consider. Suppose $A$ makes at most $cn/\gamma$ comparisons per round($c\geq 1$ is some constant.) and $A$ has $r  = o(\log^*(n)) < \log^*(n)/(2C)$ rounds. Here $C = 2^{8c}$. We will show that $A$ fails to solve \select\ with probability at least $1/3$. 

 Let's prove the following claim by induction:
\begin{claim}
After round $i$, with probability $1-i/n$, there are at least $\frac{n}{C\uparrow \uparrow i}$ elements whose comparison results are all erased. ($\uparrow\uparrow$ denotes the power tower.)
\end{claim}
\begin{proof}

Assume the claim is true for round $i-1$. Let $S$ be the set of elements whose comparison results are all erased before round $t$. By induction hypothesis, with probability $1-(i-1)/n$, $|S| \geq \frac{n}{C\uparrow \uparrow (i-1)}$. In round $i$, by averaging argument, at least $\frac{n}{2 \cdot(C\uparrow \uparrow (i-1))}$ elements in $S$ have at most $\frac{2c \cdot(C\uparrow \uparrow (i-1))}{\gamma}$ comparisons in round $i$. For each such element $s$, the probability that $s$'s comparison results are all erased in round $i$ is at least
\[
(1-\gamma)^{\frac{2c \cdot(C\uparrow \uparrow (i-1))}{\gamma}}\geq  2^{-4c \cdot(C\uparrow \uparrow (i-1))} .
\]
Let $S'$ be the set of elements whose comparison results are all erased after round $i$. Then we have
\[
\E[|S|] \geq 2^{-4c \cdot(C\uparrow \uparrow (i-1))}  \cdot \frac{n}{2 \cdot(C\uparrow \uparrow (i-1))} \geq \frac{n}{2^{6c \cdot(C\uparrow \uparrow (i-1))} }.
\]
It's easy to check that $ \frac{n}{2^{8c \cdot(C\uparrow \uparrow (i-1))} }\geq \frac{n}{C\uparrow \uparrow i} \geq \frac{n}{C\uparrow\uparrow r} \geq \sqrt{n}$. Then by Chernoff bound, 
\[
Pr[|S'| \leq  \frac{n}{2^{8c \cdot(C\uparrow \uparrow (i-1))} }] \leq e^{-\Omega(\sqrt{n})} \leq 1/n. 
\]
By union bound, with probability $1-(i-1)/n - 1/n = 1-i/n$, there are at least $\frac{n}{C\uparrow \uparrow i}$ elements whose comparison results are all erased.  
\end{proof}
Now by this claim, we know that after round $r$, with probability $1-r/n$, there are still  $\frac{n}{C\uparrow \uparrow r } \geq \sqrt{n}$ elements whose comparison results are all erased.  Let $s$ be an element whose comparison results are all erased. Even given all the relative orders of other elements except $s$, the rank of $s$ is distributed uniformly randomly in $\{1,...,n\}$. So in this case any algorithm fails to output the median with probability at least $1/2$. To sum up, any $r$-round algorithm with $O(n/\gamma)$ comparisons per round fails to output the median with probability at least $\frac{1}{2} \cdot (1- r/n) \geq 1/3$. 

\end{prevproof}
\subsubsection{Lower Bound on Rank in Noisy Model}
First, we show that \rank\ is easier than \select\ in the noisy model. We will also formally define \rank\ in the proof.

\begin{proposition}\label{prop:reduce}
{In the noisy model,} if there is an algorithm obtaining a $t$-approximation for \select\ on inputs of size $2n+2t + 1$ with probability $p$ that has query complexity $Q$ and round complexity $r$, then there is an algorithm obtaining a $t$-approximation for \rank\ on inputs of size $n$ with probability $p$ that has query complexity at most $2Q$ and round complexity $r$.
\end{proposition}
\begin{prevproof}{Proposition}{prop:reduce}
Let $A$ be the algorithm obtaining a $t$-approximation for \select\ on inputs of size $2n+2t + 1$ with probability $p$ that has query complexity $Q$ and round complexity $r$. We will first show algorithm $B$ obtaining a $t$-approximation for \rank' based on $A$. And then we will show an algorithm $C$ obtaining a $t$-approximation for \rank\ based on $B$. \rank'$(n,t)$ is defined as the following problem:
\begin{enumerate}
\item There are $n$ elements $a_1,...,a_n$, and another special element $b$. $a_1,...,a_n$ and $b$ have some underlying order unknown to the algorithm.
\item The goal is to output the number of elements in $\{a_1,...,a_n\}$ that are less than $b$ with additive error at most $t$.  
\item The algorithm is allowed to make comparisons in the noisy model between pairs $(a_i,a_j)$ and $(a_i,b)$. 
\end{enumerate}
The algorithm $B$ to solve \rank'$(n,t)$ by using $A$ is the following:
\begin{enumerate}
\item Generate elements $b_1,...,b_{n+2t+1}$.
\item Run $A$ on elements $a_1,...,a_n,b_1,...,b_{n+2t+1}$.
\item Whenever $A$ asks for a comparison:
\begin{enumerate}
\item If the comparison is $(b_i,b_j)$, wlog let $i < j$, then return $b_j$ beats $b_i$ with probability $1/2 +\gamma$ and $b_i$ beats $b_j$ with probability $1/2 -\gamma$.  
\item If the comparison is $(a_i,b_j)$, make a comparison $(a_i,b)$ and return the result to $A$.  
\item If the comparison is $(a_i,a_j)$, make a comparison $(a_i,a_j)$ and return the result to $A$.
\end{enumerate} 
\item If $A$'s output is some $b_i$, output $n + t +1 - i$. 
\item If $A$'s output is some $a_i$, fail. 
\end{enumerate}

We prove the following lemma to show that $B$ works. 
\begin{lemma}
$B$ outputs the number of elements in $\{a_1,...,a_n\}$ that are less than $b$ with additive error at most $t$ with probability at least $p$.
\end{lemma}
\begin{proof}
Suppose the underlying order of $a_1,...,a_n$ and $b$ is $a_{i_1} < a_{i_2} < \cdots a_{i_s} < b < a_{i_{s+1}} < \cdots < a_{i_n}$. Then it's easy to see from $B$ that $A$ is running on the order $a_{i_1} < a_{i_2} < \cdots a_{i_s} < b_1 < \cdots < b_{n+2t+1} < a_{i_{s+1}} < \cdots < a_{i_n}$. So the median of $a_1,...,a_n,b_1,...,b_{n+2t+1}$ will be $b_{n+t +1 -s}$. And the elements that are at most $t$ elements away from median will be $b_{n-s + 1},b_{n-s+2},...,b_{n-s + 2t + 1}$. Therefore, whenever $A$ succeeds, $B$ also succeeds. So $B$ succeeds with probability at least $p$. 
\end{proof}

Next we are going to show an algorithm $C$ obtaining a $t$-approximation for \rank\ based on $B$. Let's formally define \rank$(n,t)$ as the following problem:
\begin{enumerate}
\item There are $n$ elements $a_1,...,a_n$, and another special element $b$. There's an underlying order for each pair $(a_i,b)$. 
\item The goal is to output the number of elements in $\{a_1,...,a_n\}$ that are less than $b$ with additive error at most $t$.  
\item The algorithm is allowed to make noisy comparisons in the form $(a_i,b)$. 
\end{enumerate}
The number of comparisons used in $C$ will be at most twice the number of comparisons used in $B$. The main idea of $C$ is to create an underlying order of $a_1,...,a_n$ without knowing whether each $a_i$ is less than $b$ or not. Here is the algorithm $C$:
\begin{enumerate}
\item Run $B$ on $a_1,...,a_n$ and $b$. 
\item Whenever $B$ asks a comparison:
\begin{enumerate}
\item If the comparison is $(a_i, b)$, just make the same comparison and return the result to $B$. 
\item If the comparison is $(a_i, a_j)$, make two comparisons $(a_i, b)$ and $(a_j,b)$. 
\begin{enumerate}
\item If $a_i$ beats $b$ and $b$ beats $a_j$, then return $a_i$ beats $a_j$ to $B$. 
\item If $b$ beats $a_i$ and $a_j$ beats $b$, then return $a_j$ beats $a_i$ to $B$. 
\item If $b$ beats both $a_i$ and $a_j$, then return $a_i$ beats $a_j$ if $i > j$ to $B$. Otherwise return $a_j$ beats $a_i$ to $B$.
\item If $b$ beats both $a_i$ and $a_j$, then return $a_i$ beats $a_j$ if $i < j$ to $B$. Otherwise return $a_j$ beats $a_i$ to $B$.
\end{enumerate}
\end{enumerate}
\item Output $B$'s output.
\end{enumerate}

We prove the following lemma to show that $C$ works. 
\begin{lemma}
$C$ outputs the number of elements in $\{a_1,...,a_n\}$ that are less than $b$ with additive error at most $t$ with probability at least $p$.
\end{lemma}
\begin{proof}
Here we only have to show that $C$ feeds $B$ with the noisy comparison results that are consistent with some underlying order of $a_1,...,a_n$ and $b$, and this order is also consistent with the underlying orders of all pairs $(a_i,b)$ in $C$'s input.  

Let's look at the comparison results fed to $B$. There are several cases. Work through these case we will see what is the underlying order, and thus prove this lemma:
\begin{enumerate}
\item If the comparison is in the form $(a_i,b)$, definitely the above property is satisfied.
\item If the comparison is in the form $(a_i, a_j)$ and $a_i < b$, $a_j< b$, wlog let $i<j$, then the probability that $C$ returns $a_i$ beats $a_j$ to $B$ is $(1/2-\gamma)^2 + (1/2+\gamma)(1/2-\gamma) = 1/2- \gamma$. So for $a_i$'s that are less than $b$, the order the same as the order of indices. 
\item If the comparison is in the form $(a_i, a_j)$ and $a_i > b$, $a_j >  b$, wlog let $i<j$, then the probability that $C$ returns $a_i$ beats $a_j$ to $B$ is $(1/2+\gamma)^2 + (1/2+\gamma)(1/2-\gamma) = 1/2+ \gamma$. So for $a_i$'s that are greater than $b$, the order the same as the reverse order of indices. 
\item If the comparison is in the form $(a_i, a_j)$ and $a_i < b < a_j$,  then the probability that $C$ returns $a_i$ beats $a_j$ to $B$ is $(1/2+\gamma)^2 + (1/2+\gamma)(1/2-\gamma) = 1/2 -  \gamma$. So if $a_i < b < a_j$, no matter $i < j$ or $i > j$, $C$ always feeds $B$ that $a_i$ beats $a_j$ with probability $1/2 - \gamma$. 
\end{enumerate}
\end{proof}
\end{prevproof}

\begin{prevproof}{Theorem}{thm:lbnoisy}
Let $t = n^{3/8}/40$. Suppose we have an algorithm $C$ for \rank\ with $o(n\log n/\gamma^2) < \frac{n\log n}{800\gamma^2}$ comparisons. We are going to show that $C$ obtains a $t$-approximation for \rank\ with probability less than $2/3$. Let's consider $C$ on the following input distribution: each $a_i$ is independently chosen to be less than $b$ with probability $1/2$. Since we are considering the success probability of $C$ on a distribution, we can wlog assume $C$ is deterministic. 

Then we are going to use the technique similar to the technique in the proof of Theorem \ref{thm:findMinlb}. Let $S$ be the comparison results. Let $g(i)$ denote the number of correct comparisons between $a_i$ and $b$. Let $b(i)$ denote the number of wrong comparisons between $a_i$ and $b$. Let $W$ be the event that there are at least $n/4$ $i$'s such that $|g(i) - b(i)| \leq \frac{\log n}{64\gamma}$. Notice $W$ can be directly observed from $S$. Similarly as the proof of Theorem \ref{thm:findMinlb}, we need two new terms $p(i)$ and $q(i)$. They are set to be 0 before the algorithm's run and we will show how they change during the algorithm. For each comparison between $a_i$ and $b$, let the followings happen: 
\begin{enumerate}
\item With probability $\gamma$, $p(i)$ is increased by 1.
\item With probability $(1-\gamma)/2$, $q(i)$ is increased by 1. 
\item With probability $(1-\gamma)/2$, $q(i)$ is decreased by 1. 
\end{enumerate}
It's easy to check that $p(i) + q(i)$ has exactly the same distribution as $g(i) - b(i)$.  From the rule above, we have 
\[
\E \sum_{i=1}nm p(i) = \frac{n \log n }{800\gamma}
\] and 
\[
\E \sum_{i=1}^n q(i)^2  = \frac{(1-\gamma) n \log n}{800\gamma^2}.
\]
By Markov inequality, with probability at least $3/4$, $\sum_{i=1}^n p(i) \leq  \frac{ n\log n }{200\gamma}$. With probability at least $3/4$, $\sum_{i=1}^n q(i)^2 \leq \frac{n\log n}{200\gamma^2}$. By union bound, with probability at least $1/2$, $\sum_{i=1}^n p(i) \leq  \frac{ n\log n }{200\gamma}$ and $\sum_{i=1}^n q(i)^2 \leq \frac{n\log n}{200\gamma^2}$. In this case, at least $n/2$ $i$'s satisfy $p(i) \leq \frac{\log n }{100\gamma}$ and at least $3n/4$ $i$'s satisfy $q(i)^2 \leq  \frac{\log n}{50\gamma^2}$. Then at least $n/4$ $i$'s satisfy both  $p(i) \leq \frac{\log n }{100\gamma}$ and $q(i)^2 \leq  \frac{\log n}{50\gamma^2}$. Then for these $i$'s we have,
\[
|p(i) + q(i)| \leq p(i) + |q(i)| \leq \frac{\log n }{100\gamma} +\frac{\sqrt{\log n} }{\sqrt{50} \cdot \gamma}\leq \frac{\log n}{64\gamma}. 
\]

Now condition on some comparison results $S$ such that $W$ is true. Let $X$ be the set of $i$'s such that $|g(i) - b(i)| \leq  \frac{\log n}{64\gamma}$. Let $<_{\bar{X}}$ be some underlying order between each pair $(a_i,b)$, $i \not \in X$. We also condition on $<_{\bar{X}}$, and we want to study the distribution of underlying orders between pairs $(a_i, b)$, $i \in X$. One critical thing to notice here is that since each $a_i$ is only compared with $b$, even condition on the comparison results $S$, whether $a_i < b$ is independent from whether $a_j < b$. For each $i \in X$, we have
\[
\frac{Pr[a_i < b | S, <_{\bar{X}}]}{Pr[a_i > b | S, <_{\bar{X}}]} = \frac{Pr[a_i < b | S]}{Pr[a_i > b | S]} \leq (\frac{1/2 + \gamma}{1/2 - \gamma})^{\frac{\log n}{64\gamma}} \leq e^{8\gamma \cdot \frac{\log n}{64\gamma}} \leq 2^{\frac{\log n}{4}} = n^{1/4}.
\]
Similarly, we also have 
\[
\frac{Pr[a_i < b | S, <_{\bar{X}}]}{Pr[a_i > b | S, <_{\bar{X}}]}  \geq \frac{1}{n^{1/4}}.
\]
Put them together, we get
\[
\frac{1}{2n^{1/4}} \leq Pr[a_i < b | S, <_{\bar{X}}] \leq 1 - \frac{1}{2n^{1/4}}.
\]
Therefore $Pr[W] \geq 1/2$. 

Finally we are going to use the anti-concentration inequality to show that algorithm $C$  succeeds with probability at most $1/5$ when $W$ holds. Let $u_i = Pr[a_i < b| S,<_{\bar{X}}]$. We have $\E[X_i] = 0$. Let $X_i = 1_{a_i < b|S , <_{\bar{X}}} - u_i$. Let $\sigma^2_i = \E[X_i^2] = (1-u_i)u_i > \frac{1}{4n^{1/4}}$.  Let $\rho_i = \E[|X_i|^3]$. Let $O = \sum_{i \in X} 1_{a_i < b|S , <_{\bar{X}}}$. Let $F(x)$ be the cumulative distribution function of $\frac{O}{\sqrt{\sum_{i \in X} \sigma_i^2}}$. Let $\Phi(x)$ be the cumulative distribution function of standard normal distribution. By Berry-Eseen inequality, we have
\[
\sup_{x \in \mathbb{R}} |F(x) - \Phi(x) | \leq (\sum_{i\in X} \sigma^2_i)^{-1/2} \cdot \max_{i \in X} \frac{\rho_i}{\sigma_i^2}  = \leq (\frac{m}{100} \cdot \frac{1}{4n^{1/4}})^{-1/2} \cdot \max_{i \in X}((1-u_i)^2 + u_i^2) \leq 4n^{-3/8}.
\] 
Let $C(S)$ be the output of the algorithm, we have
\begin{eqnarray*}
Pr[C \text{ succeeds}| S, <_{\bar{X}}] &\leq& F(\frac{C(S) + t - \sum_{i \in X} u_i}{\sqrt{\sum_{i \in X} \sigma_i^2}}) - F(\frac{C(S) - t - \sum_{i \in X} u_i}{\sqrt{\sum_{i \in X} \sigma_i^2}}) \\
&\leq& \Phi(\frac{C(S) + t - \sum_{i \in X} u_i}{\sqrt{\sum_{i \in X} \sigma_i^2}}) - \Phi(\frac{C(S) - t - \sum_{i \in X} u_i}{\sqrt{\sum_{i \in X} \sigma_i^2}})  + 8n^{-3/8}\\
&\leq& \Phi(\frac{t }{\sqrt{\sum_{i \in X} \sigma_i^2}}) - \Phi(\frac{ - t}{\sqrt{\sum_{i \in X} \sigma_i^2}})  + 8n^{-3/8}\\
&\leq& \Phi(1/10) - \Phi(-1/10) + 8n^{-3/8} < 1/5. \\
\end{eqnarray*}
To sum up,
\[
Pr[C \text{ succeeds}] \leq Pr[\neg W] + Pr[C \text{ succeeds}| W] < Pr[\neg W] + Pr[W]  \cdot (1/5) \leq 1/2 + (1-1/2) \cdot (1/5) \leq 2/3. 
\]

\end{prevproof}

\end{document}